\documentclass[aos, preprint]{imsart}

\RequirePackage[OT1]{fontenc}
\RequirePackage{amsthm,amsmath,amssymb}
\RequirePackage[numbers]{natbib}
\RequirePackage[colorlinks,citecolor=blue,urlcolor=blue]{hyperref}

\usepackage{multirow}
\usepackage{amsmath}
\allowdisplaybreaks[1]

\usepackage{lscape}
\usepackage{graphicx}
\usepackage{psfrag}
\usepackage[colorlinks]{hyperref}
\usepackage{enumitem}
\usepackage{subfig}
\graphicspath{{Simulations/SV/Figures/}, {Simulations/CIR/Figures/}}

\startlocaldefs
\numberwithin{equation}{section}
\theoremstyle{plain}
\newtheorem{theorem}{Theorem}[section]
\newtheorem{lemma}[theorem]{Lemma}
\newtheorem{assumption}{Assumption}
\newtheorem{remark}[theorem]{Remark}
\newtheorem{cor}[theorem]{Corollary}
\newtheorem{prop}[theorem]{Proposition}
\newtheorem{defi}[theorem]{Definition}
\newcommand{\pg}[2]{g_{\xi_{#1}}(#2)}
\def\myi{i}
\def\ci{(\myi-1)h}
\def\ni{\myi h}

\def\nd{d}
\def\d{{\rm d}}

\def\lp{\left(}
\def\rp{\right)}
\def\lsb{\left[}
\def\rsb{\right]}
\def\lcb{\left\{}
\def\rcb{\right\}}
\def\mE{{\mathbb E}}
\def\Un{\mathrm {Un}}
\def\mR{{\mathbb R}}
\def\mZ{{\mathbb Z}}
\def\mP{{\mathbb P}}
\def\mN{{\mathbb N}}
\def\one{{\bf 1}}
\def\trace{{\rm tr}}
\def\ESS{{\rm ESS}}
\def\sign{{\rm sign}}

\def\ind{{\bf 1}}
\def\CIS{{\rm CIS}}

\newcommand{\bb}{b}

\newcommand{\bsigma}{\sigma}
\newcommand{\B}{B}

\def\dc{\sigma}
\def\dm{\gamma}
\def\gd{{\cal N}}
\def\one{{\bf 1}}

\def\pin{\delta}
\def\pex{\alpha}

\def\th{T}

\def\artime{\tau}
\newcommand{\iat}[1]{\artime_{#1}}

\newcommand{\te}[1]{n_{#1}}
\newcommand{\lat}[1]{\artime_{#1}}

\newcommand{\dt}[1]{\Delta \artime_{#1}}

\def\algw{w}
\def\rvw{W}
\newcommand{\w}[1]{\algw_{\iat{#1}}}
\newcommand{\ww}[1]{\algw_{#1}}

\newcommand{\WW}[1]{\rvw_{#1}}

\def\algx{x}
\def\rvx{X}
\def\bx{{\bar \algx}}
\newcommand{\x}[1]{\algx_{\iat{#1}}}
\newcommand{\xx}[1]{\algx_{#1}}
\newcommand{\X}[1]{\rvx_{\iat{#1}}}
\newcommand{\XX}[1]{\rvx_{#1}}

\newcommand{\iw}[1]{H_{#1}}

\def\ck{k} \def\istart{0} \def\kspace{\mN} \def\kstart{0}
\def\nk{k+1} \def\iend{\te{\th}-1}

\newcommand{\prt}[3]{#1_{#2}^{(#3)}}

\def\algsis{1}
\def\algrwsis{2}
\def\algrbsis{3}
\def\algcis{4}
\def\algcc{5}
\def\algcisra{6}
\def\algcisrb{7}

\newcommand{\y}{y}

\def\bder{\bb_{(1)}}
\def\dmder{\dm_{(1)}}
\def\dmdder{\dm_{(2)}}

\newcommand{\ecis}{\mathbb{E}_{\textup{\texttt{CIS}}}}

\newcommand{\trd}[3]{p(#1,#2,#3)}

\newcommand{\wtrd}[3]{p^{(1)}(#1,#2,#3)}

\newcommand{\ttrd}[3]{{\tilde p}(#1,#2,#3)}

\newcommand{\prd}[4]{q_{#4}(#1,#2,#3)}

\def\tro{{\mathcal K}}

\newcommand{\pro}{{\mathcal K}_{\theta}}
\newcommand{\prot}[1]{{\mathcal K}_{#1}}

\newcommand{\rw}[5]{r_{#5}(#1, #2, #3, #4)}

\newcommand{\mrwh}[5]{{\bar r}_{#5}(#1, #2, #3, #4)}
\newcommand{\mrw}[4]{\rho_{#4}(#1, #2, #3)}

\newcommand{\mrwcc}[4]{\rho_{\theta_{#1}}(#2, #3, #4)}

\newcommand{\prdcc}[4]{q_{\theta_{#1}}(#2,#3,#4)}
\newcommand{\prdccn}[4]{q_{#1}(#2,#3,#4)}

\def\qder{\Lambda}
\newcommand{\qderf}[4]{\Lambda_{#1}(#2, #3, #4)}
\newcommand{\qderft}[4]{\Lambda^{T}_{#1}(#2, #3, #4)}

\def\qdder{K}
\newcommand{\qdderf}[4]{K_{#1}(#2, #3, #4)}

\endlocaldefs

\begin{document}

\begin{frontmatter}

\title{Continuous-time Importance Sampling: Monte Carlo Methods which Avoid Time-Discretisation Error}
\runtitle{Continuous-time Importance Sampling}

\begin{aug}
\author{\fnms{Paul} \snm{Fearnhead}\corref{}\thanksref{m1}\ead[label=e1]{p.fearnhead@lancaster.ac.uk}},
\author{\fnms{Krzystof} \snm{\L{}atuszy\'nski}\thanksref{m2}\ead[label=e2]{k.g.latuszynski@warwick.ac.uk}},
\author{\fnms{Gareth O.} \snm{Roberts}\thanksref{m2}\ead[label=e3]{gareth.o.roberts@warwick.ac.uk}}
\and
\author{\fnms{Giorgos} \snm{Sermaidis}\thanksref{m1}\ead[label=e4]{g.sermaidis@lancaster.ac.uk}}
\address{\printead{e1}}

\runauthor{P. Fearnhead et al.}

\affiliation{Lancaster University\thanksmark{m1} and University of Warwick\thanksmark{m2}}

\address{P. Fearnhead\\
Department of Mathematics and Statistics\\
Fylde College\\
Lancaster University\\
Lancaster LA1 4YF\\
United Kingdom\\
\printead{e1}}

\address{Krzystof \L{}atuszy\'nski\\
Department of Statistics\\
University of Warwick\\
Coventry CV4 7AL\\
United Kingdom\\
\printead{e2}}

\address{G. O. Roberts\\
Department of Statistics\\
University of Warwick\\
Coventry CV4 7AL\\
United Kingdom\\
\printead{e3}}

\address{G. Sermaidis\\
Department of Mathematics and Statistics\\
Fylde College\\
Lancaster University\\
Lancaster LA1 4YF\\
United Kingdom\\
\printead{e4}}

\end{aug}

\begin{abstract}
\end{abstract}



\end{frontmatter}

\tableofcontents

\section{Introduction}\label{S:Intro}

Continuous-time Markov processes (CTMPs) are used extensively for modelling
random processes which evolve continuously in time. They have found
numerous applications in many diverse scientific areas, such as
physics \cite{van1992stochastic, MR987631}, chemistry
\cite{van1992stochastic}, finance \cite{MR1680267, duffie2010dynamic,
  MR2042661}, insurance \cite{MR1680267} and
systems biology \cite{MR2222876}.  The goal of this paper is to introduce a novel algorithm for
simulating from CTMPs and develop
Monte Carlo (MC) estimators of expectations of functionals of CTMPs
which arise frequently in many practical applications.

A specific class of CTMPs which has motivated the work in this paper
is diffusion processes. A diffusion process $X\in\mR^\nd$ is defined
as the solution to a stochastic differential equation (SDE),
\begin{align} \label{eq:SDE} \d\XX{t} = \bb(\XX{t})\d
  t+\bsigma(\XX{t})\d\B_t,\quad \XX{0}=\xx{0},\, t\in[0,\th],
\end{align}
where $\B$ is a standard $\nd$-dimensional Brownian motion. The
functionals $\bb:\mR^\nd\to\mR^\nd$ and
$\bsigma:\mR^\nd\to\mR^{\nd\times\nd}$ are known as the drift and
diffusion coefficient respectively and model the mean and variance of
the infinitesimal increments of the process. We assume that the
functionals satisfy growth and Lipschitz conditions, so that
\eqref{eq:SDE} admits a unique weak non-explosive solution
\citep{oksendal}.



Simulation and inference for diffusion processes is a very challenging
task due to the intractability of the transitional dynamics of the
process. 
Most contemporary
approaches are based on approximating the diffusion by a discrete-time
Markov process. If we choose $M\in\mZ^{+}$ then we can partition a
time interval $[0,T]$ into $M$ subintervals of size  $h:=T/M$. 
The dynamics of the SDE can then be approximated by those of a discrete-time
Markov process, with
\begin{align}\label{eq:euler_scheme}
  \XX{t_{i+1}}=\XX{t_{i}} + h\bb(\XX{t_i}) +
  \sqrt{h}\bsigma(\XX{t_i})Z_{i+1},
\end{align}
where $t_i:=ih,\,i=0,1,\ldots,M$, and
$Z_i$s are independent standard Gaussian random variables. The above
scheme, known as the Euler scheme, provides a straightforward way for
simulating the target process and allows us to design MC estimators
using independent draws from~\eqref{eq:euler_scheme}. However, the
discrete approximation of the continuous-time model introduces a
systematic error which is eliminated only when $M\to\infty$. The value
of $M$ that guarantees a sufficiently accurate approximation is
typically unknown, and is found empirically by multiple runs of the MC
estimators over increasing values of $M$ until no significant change is
shown. Alternatively, $M$ is selected as large as the computational
resources allow, which may result in discretisations finer than
necessary. Both selection procedures induce a significant
computational burden and waste of resources. Furthermore, the
approximation error of such MC estimators is hard to quantify;
available results only provide upper bounds for the bias and these
are typically known only up to a constant.

More recently, an unbiased approach was introduced with the
development of the exact algorithm
(EA)\citep{besk:papa:robe:fear:2006,besk:papa:robe:2008}. The EA is a
rejection sampling algorithm on the diffusion path space which
returns, in principle, a continuous trajectory from \eqref{eq:SDE}. In
practice, the output is an exact draw from the finite-dimensional
distribution of a ``skeleton'' for the diffusion path. 
This skeleton can then be used to evaluate the diffusion path at any
time instance with no further reference to the target process. The EA
allows the design of MC estimators whose only source of error is due
to MC, and thus easier to quantify. In particular, for a fixed
computational cost $K$, the approximation error of the unbiased
methods is of order ${\cal O}(K^{-1/2})$, in contrast to the rate
${\cal O}(K^{-1/3})$ which is usually attained Euler based approaches
 (though see \cite{Giles} for some interesting improvements on this). 

Unfortunately, the
range of applicability of the EA is limited by two conditions. First,
the existence of a variance-stabilising transformation, known as the
Lamperti transform, after which the process has unit diffusion matrix,
and second, the drift of the transformed process has to be of gradient
form. The reason that these conditions are needed for the EA is that the EA is
based upon rejection and importance sampling methods that work on the path-space 
of an SDE. These require tractable proposal processes whose law for the
diffusion path is absolutely continuous with the SDE of interest. In general such
tractable proposal processes only exist if the SDE of interest can, for example,
be transformed to unit diffusion matrix by the Lamperti transform.
Both conditions for the EA are usually satisfied in a univariate setting,
but pose serious limitations for general multivariate processes, as
illustrated in \cite{sah:2008}.


In this paper we develop a continuous-time sequential importance
sampling (CIS) algorithm which eliminates time-discretisation errors
and provides online unbiased estimation for diffusions. Our work
removes the strong conditions imposed by the EA and thus extends
significantly the class of discretisation error-free MC methods for
diffusions. The reason that CIS can be applied more generally than EA
is that it no longer works on the path space of the SDE. Instead it uses 
proposal distributions for the transition density of the diffusion, and proposal
distributions that are absolutely continuous with respect to the true transition
density exist for general SDEs.

The CIS algorithm is built by modifying a standard
discrete-time sequential importance sampling (SIS) algorithm by
introducing updates at random times according to a Bernoulli process
and replacing intractable importance weights by unbiased estimators,
similarly to \cite{fear:papa:robe:2008}. The CIS algorithm is then
obtained by taking the continuous-time limit of the discrete SIS. The
key is that in the limit the Bernoulli process converges to a Poisson
process implying that updates occur only at a finite number of
points. Therefore, in principle, CIS evolves continuously in time but is
implemented using only finite computation. Furthermore, the sequential
nature of our approach allows us to boost the algorithmic performance
of CIS. By using batch implementations and standard resampling
techniques we develop a sequential importance resampling algorithm and
demonstrate empirically that it can lead to significantly more
accurate estimates than the plain version of CIS when the time horizon~$T$ is large.

Whilst motivated for sampling from SDEs, the CIS can be applied to other CTMPs. 
One important application is to use CIS to simulate from CTMPs that have been designed to have a particular stationary distribution of interest. The results in this paper provide the underpinning theoretical justification for some of the extensions of the recently proposed SCALE algorithm of \cite{pollock2016scalable} that are suggested in \cite{fearnhead2016piecewise}.


The paper is structured as follows. The next section gives an intuitive derivation of the 
CIS algorithm in terms of a continuous-time limit of standard discete-time sequential importance sampling algorithms.
Section~\ref{sec:setup-cis}
then formally defines the generic version of the CIS algorithm along with the
resulting MC estimators. Unbiasedness and stability of the generic CIS is
addressed in Section~\ref{sec:validity-cis}. Sections~\ref{sec:CIS_diff} and~\ref{sec:validity-diffusions} provide the CIS version for diffusions
and detailed treatment
of its stability respectively. Section~\ref{sec:extensions-cis} deals
with boosting the algorithms performance and 
presents two modifications based on optimal proposals and batch
implementation with resampling. Section~\ref{sec:numericals} assesses the performance of
our MC estimators and compares it to several existing
methods. Section~\ref{sec:discussion} concludes with a discussion.









\subsection{Acknowledgements}

The work is partially funded by EPSRC grants EP/G028745/1 and EP/K014463/1. K{\L}
acknowledges funding form the Royal Society via the University
Research Fellowship scheme.  
We thank Andreas Eberle and Omiros Papaspiliopoulos for useful
comments at various stages of the project and Arturo Kohatsu-Higa
for an insightful discussion of his work \cite{bally2015probabilistic}.
\section{Intuitive derivation: CIS as a limiting algorithm} \label{S:deriv}

To help understand the idea behind CIS, we first give an intuitive derivation of the CIS
algorithm as the limiting algorithm of a standard discrete-time
sequential importance sampler. Our target process is a CTMP with transition density 
\[
\trd{x}{\y}{t} = \mathbb{P}(X_{s+t}\in\mbox{d}\y|X_s=x)/\mbox{d}\y~~\mbox{
  for $s,t>0$},
\]
with $X_0=x_0$. For the importance sampler we will use a family of proposal distributions that
approximate this transition density. Denote these by $\prd{x}{y}{t}{\theta}$, where $\theta$ is 
the parameter of the family of distributions. 

\subsection*{Standard Sequential Importance Sampling}

We first describe a standard discrete-time sequential importance
sampler \citep{kong:liu:1994}. Firstly discretise the time interval
$[0,\th]$ into intervals of length $h$. The sequential importance
sampler will simulate values of the process and importance weights at
times $h,2h,\ldots$. If at time $\ci$ we have a value and weight,
$\xx{\ci},\ww{\ci}$, then we first generate $\xx{\ni}$ from
$\prd{\xx{\ci}}{\cdot}{h}{\theta_{i-1}}$, for some suitably chosen $\theta_{i-1}$, and set
$$w_{\ni}=w_{\ci}\trd{\xx{\ci}}{\xx{\ni}}{h}/\prd{\xx{\ci}}{\xx{\ni}}{h}{\theta_{i-1}}.$$ 
Full details are given in Algorithm \algsis.

\begin{description}

\item[{\bf Algorithm \algsis}:] {\it Discrete-time Sequential Importance Sampler} (SIS).

\item[{\it Input}:] $\th$ a time to stop, $M$ a discretisation level and
  $\xx{0}$ the initial value.

\item[{\it Initialise}:] Set $w_0=1$ and $h=\th/M$.
  
\item[{\it Iterate}:] For $1\leq \myi\leq M$
    \begin{enumerate}
    \item choose $\theta_{i-1}$ and draw $\xx{\ni}\sim \prd{\xx{\ci}}{\cdot}{h}{\theta_{\myi-1}}$,
    \item update the weight
      \[
      w_{\ni}=w_{\ci}\frac{\trd{\xx{\ci}}{\xx{\ni}}{h}}{\prd{\xx{\ci}}{\xx{\ni}}{h}{\theta_{\myi-1}}}.
      \]
    \end{enumerate}
  
\end{description}

Let $X_{\myi h},W_{\myi h}$ denote the random variables of the value and
weight at time $\myi h$, respectively. Standard importance sampling
arguments show that whenever the left-hand side is defined, then
\begin{equation} \label{eq:2}
\mathbb{E}_p[f(X_{\myi h})]=\mathbb{E}_q[W_{\myi h}f(X_{\myi h})],
\end{equation}
where $\mathbb{E}_q$ denotes expectation of $X_{\myi h},W_{\myi h}$ with
respect to the sequential importance sampling process.

Unfortunately Algorithm {\algsis} cannot be implemented
since $\trd{x}{\y}{h}$ is unknown, hence the weights cannot be
updated at each iteration. To overcome this, the aim is to consider
the limit as $h\rightarrow 0$. This cannot be considered directly as such a limit would not make sense. So first we adapt
Algorithm {\algsis} using in turn ideas of random weight sequential importance
sampling \citep{fear:papa:robe:2008},
retrospective sampling \citep{papa:robe:2008} and
Rao-Blackwellisation \cite[]{liu:chen:1998}.

\subsection*{Random Weight sequential Importance Sampling}

Our first idea is to change how the weights are updated, and to make
this update random. This will still produce a valid importance
sampler, for which (\ref{eq:2}) is satisfied, provided that the mean
of the random weight update is equal to the true incremental weight.

We will implement our random weight update by introducing a series of
Bernoulli random variables, $U_{\myi h}$, for $\myi =1,2,\ldots$.  We call
successes of the Bernoulli random variables, events. We will allow the
distribution of $U_{\myi h}$ to depend on the time $\myi h$ and also the
history of the Bernoulli process, but for concreteness and notational
simplicity we will consider only the case where it depends on the time
since the most recent event has occurred. That is, if the most recent
event has happened at time $jh$, then we assume there exists a
function $\lambda(\cdot)>0$ such that with probability
$h\lambda(\ni-jh)$, $U_{\ni}=1$, otherwise $U_{\ni}=0$. It is trivial to
see that if
\begin{align*}
\rw{x}{\y}{h}{u}{\theta}:=1+\frac{1}{h\lambda(u)}\left\{\frac{\trd{x}{y}{h}}{\prd{x}{\y}{h}{\theta}}-1\right\},
\end{align*}
then
\begin{align*}
  \frac{\trd{\xx{\ci}}{\xx{\ni}}{h}}{\prd{\xx{\ci}}{\xx{\ni}}{h}{\theta_{\myi-1}}}=\mathbb{E}\lsb
    (1-U_{\ni})+U_{\ni}\rw{\xx{\ci}}{\xx{\ni}}{h}{(\myi-j)h}{\theta_{\myi-1}}\rsb.
\end{align*}
Thus, we can define an appropriate random incremental weight that
takes the value $1$ if $U_{\myi h}=0$ and
$\rw{x_{\ci}}{x_{\ni}}{h}{(\myi-j)h}{\theta_{\myi-1}}$ if $U_{\ni}=1$. The
resulting algorithm is shown in Algorithm {\algrwsis}.

\begin{description}

\item[{\bf Algorithm \algrwsis}:] {\it Random weight SIS}.

\item[{\it Input}:] $\th$ a time to stop, $M$ a discretisation level and $\xx{0}$ the initial value.

\item[{\it Initialise}:] Set $\ww{0}=1,\, h=\th/M$ and $j=0$.
  
\item[{\it Iterate}:] For $1\leq \myi\leq M$
  \begin{enumerate}
  \item choose $\theta_{\myi-1}$ and draw $\xx{\ni}\sim
    \prd{\xx{\ci}}{\cdot}{h}{\theta_{\myi-1}}$,
    
  \item draw $U_{\ni}$, a Bernoulli random variable with success
    probability $ h\lambda(\myi h-jh)$,
    
  \item if $U_{\ni}=1$, set
      \[
      \ww{\ni}=\ww{\ci}\rw{\xx{\ci}}{\xx{\ni}}{h}{(\myi-j)h}{\theta_{\myi-1}},
      \]
      and $j=\myi$; else set $w_{\ni}=w_{\ci}$.
  \end{enumerate}
  
\end{description}

\subsection*{Retrospective Sampling and Rao-Blackwellisation}

Algorithm {\algrwsis} cannot be implemented in practice, as the random
incremental weights still depend on the unknown transition density of
the proposal process. Furthermore, even if it could be implemented, it
would be less efficient than Algorithm {\algsis} because of the extra
randomness introduced by the random variables
$U_{\myi h},\,\myi=1,2,\ldots$. However its use is that it now enables us to
use the ideas of retrospective sampling and  Rao-Blackwellisation.

The idea behind retrospective sampling is to interchange the order of
simulating $\xx{\myi h}$ and $U_{\myi h}$, and to note that if
$U_{\myi h}=0$ then we do not need to simulate $\xx{\myi h}$ at all. So we
will only simulate $\xx{\myi h}$ at time points when there is an event. If
$jh$ is the time of the most recent event prior to $\myi h$, then
$\xx{\myi h}$ is simulated from the proposal transition density
$\prd{x_{jh}}{\cdot}{(\myi-j)h}{\theta_j}$. To update the weight at time
$\ni$ we also need to know the value of $\xx{\ci}$ and so this has
to be simulated, conditionally on $\xx{jh}$ and $\xx{\ni}$, from the
density
\begin{equation} \label{eq:xi-1}
q(\xx{\ci}\mid x_{jh},x_{\ni}):=\frac{\prd{\xx{jh}}{\xx{\ci}}{(\myi-j-1)h}{\theta_j}\prd{\xx{\ci}}{\xx{\ni}}{h}{\theta_j}}{\prd{\xx{jh}}{\xx{\ni}}{(\myi-j)h}{\theta_j}}.
\end{equation}
This can be viewed as introducing a random incremental weight. The
weight depends on the value of $\xx{\ci}$ that is simulated. The
idea behind Rao-Blackwellisation is that a more efficient algorithm
can be produced by replacing this random incremental weight by a
deterministic one that is equal to the mean of the random incremental
weight. We denote this mean by
\begin{align*} 
  \mrwh{\xx{jh}}{\xx{\ni}}{h}{(\myi-j)h}{\theta_j} = \int
  \rw{\xx{\ci}}{\xx{\ni}}{h}{(\myi-j)h}{\theta_j}q(\xx{\ci}\mid \xx{jh},\xx{\ni})
  \d \xx{\ci}.
\end{align*}
The resulting algorithm is given in detail in Algorithm
\algrbsis. If we stop the algorithm at a time point at which
there is no event, then the weight at that time is equal to the value
of the weight at the last event time, and the value of the process can
be simulated from the appropriate transition density of the proposal
process given the time and value at the last event time.

\begin{description}

\item[{\bf Algorithm \algrbsis}:] {\it SIS with Retrospective Sampling and Rao-Blackwellisation}.

\item[{\it Input}:] $\th$ a time to stop, $M$ a discretisation level and $\xx{0}$ the initial value.

\item[{\it Initialise}:] Set $w_0=1,\, h=\th/M,\,j=0$ and choose $\theta_j$.
  
\item[{\it Iterate}:] For $1\leq \myi\leq M$
  \begin{enumerate}
    
  \item draw $U_{\ni}$, a Bernoulli random variable with success
    probability $ h\lambda(\ni-jh)$,

  \item if $U_{\ni}=1$, then
    \begin{enumerate}
    \item draw $\xx{\ni}\sim
    \prd{\xx{jh}}{\cdot}{h}{\theta_j}$,
  \item update the weight
    \begin{align*}
      w_{\ni}=w_{\ci}\,\mrwh{\xx{jh}}{\xx{\ni}}{h}{(\myi-j)h}{\theta_j},
    \end{align*}
  \item set $j=\myi$ and choose $\theta_j$.
    \end{enumerate}
  \end{enumerate}
  
\end{description}

\subsection*{CIS as a limiting algorithm}

Whilst Algorithm {\algrbsis} cannot be implemented, we can now take
the limit as $h\rightarrow0$. In doing this the Bernoulli random
variables will converge to a renewal process with the rate of an event at
a time $\tau$ since the last event being $\lambda(\tau)$. We will thus
obtain a limiting algorithm where events are simulated from a renewal
process, the value of the process at these events are drawn from the
proposal transition density, and the weights are updated based on the
simulated value. The key thing is that in this limit the incremental
weight is tractable. If we fix $s=jh$ and $t=kh$ and let $h\rightarrow
0$, then
\begin{align*}
  &\lim_{h\rightarrow 0}
  \mrwh{\xx{s}}{\xx{t}}{h}{t-s}{\theta_s}=\\
&= 1 + \frac{1}{\lambda(t-s)} \lim_{h\rightarrow 0} \int
\frac{1}{h}\left\{\frac{\prd{x_{s}}{\y}{t-s-h}{\theta_s}[\trd{\y}{x_{t}}{h}-\prd{\y}{x_{t}}{h}{\theta_s}]}{\prd{x_{s}}{x_{t}}{t-s}{\theta_s}}\right\}\mbox{d}\y.
\end{align*}
However, by integrating the Kolmogorov's forward equation (c.f. (\ref{eq:PDE})) over a time interval of length
$h$, with an initial distribution for the process being $
\prd{x_{s}}{\y}{t-s-h}{\theta_s}$ we get:
\begin{align*} \int
  \prd{x_{s}}{\y}{t-s-h}{\theta_s}\trd{\y}{x_{t}}{h}\mbox{d}\y &=
  \prd{x_{s}}{x_t}{t-s-h}{\theta_s} \\
  &+ h \tro\prd{x_{s}}{x_t}{t-s-h}{\theta_s}+o(h),
\end{align*}
where $\tro$ is the forward operator of the target CTMP. A similar argument applies for the term which involves
$\prd{\y}{x_t}{h}{\theta_s}$. Thus
\begin{align*} 
  \lim_{h\rightarrow 0} \mrwh{\xx{s}}{\xx{t}}{h}{t-s}{\theta_s}
  = 1+\frac{\lsb\tro-\prot{\theta_s} \rsb
    \prd{x_{s}}{\y}{t-s}{\theta_s}|_{\y=x_t}}{\lambda(t-s)\prd{x_{s}}{x_{t}}{t-s}{\theta_s}},
\end{align*}
where $\prot{\theta_s}$ is the forward operator associated with the proposal 
transition density $\prd{\y}{x_t}{h}{\theta_s}$.
%



\section{Setup and the CIS algorithm}\label{sec:setup-cis}

This section introduces the basic notation and presents the CIS
algorithm.

\subsection{Notation}\label{section:notation}

The $i$th element of a vector $x$ is denoted by $x_i$ or $[x]_i$. The Euclidean
inner product of two vectors~$x$ and~$y$ is denoted
by \begin{eqnarray*} x\cdot y & = & \sum_i x_i y_i.\end{eqnarray*} For
a matrix~$x$, its $(i,j)$th element is denoted by~$x_{ij}$ or $[x]_{ij}$, its
transpose by~$x^T$, its inverse by~$x^{-1}$, its inverse transpose
by~$x^{-T}$, and its complex conjugate by $\bar{x}.$ An $m\times n$ matrix whose elements are
equal to one is denoted by $\one_{m\times n}$.  The Frobenius inner
product of two matrices $x$ and $y$ of the same dimension, i.e., the
sum of all entries of their element-wise product, is denoted
by \begin{eqnarray*} x:y & = &
\sum_i\sum_jx_{ij}y_{ij}. \end{eqnarray*} 
Unless stated otherwise, the Frobenius (Euclidean) matrix norm
 \begin{eqnarray*} \|x\|^2 & = & \sum_i \sum_j |x_{ij}|^2 = 
\trace(x\bar{x}^T) \end{eqnarray*} 
is used throughout the paper and when applying the norm we treat
vectors as matrices.\\
The first and second derivative of a scalar function $f:\mR\to\mR$ are
denoted by $f^{'}$ and $f^{''}$ respectively. The gradient of a scalar
function $f:\mR^\nd\to\mR$ and the Hessian matrix are written
respectively as
\begin{eqnarray*}
 [\nabla_x f (x)]_i & := & \partial f(x)/\partial x_i \\  
 \left[
H_x f(x)\right]_{ij} 
 & := & \partial^2f(x)/\partial x_i\partial x_j. 
 \end{eqnarray*} 
Furthermore, for the
drift $\bb$ and the diffusion coefficient $\dc$ in \eqref{eq:SDE} we
define \begin{eqnarray*} \dm:\mR^\nd\to\mR^{\nd\times\nd} \quad  & \textrm{as}
  & \quad \dm \; = \; \dc\dc^T, \\
\bder:\mR^\nd\to\mR^\nd  \quad  & \textrm{as}
  & \quad [\bder(x)]_i \; = \; \partial \bb_i(x)/\partial
x_i, \\ 
\dmder:\mR^\nd\to\mR^{\nd\times\nd}  \quad  & \textrm{as}
  & \quad [\dmder(x)]_{ij} \; = \;
\partial\dm_{ij}(x)/ \partial x_j, 
 \\
 \dmdder:\mR^\nd\to\mR^{\nd\times\nd} \quad  & \textrm{as}
  & \quad [\dmdder(x)]_{ij}\; = \; \partial^2\dm_{ij}(x)/ \partial x_i\partial x_j.\end{eqnarray*}
For a finite sequence $\{\iat{k}\}_{k\in\kspace},\,
0<\iat{k}<\iat{k+1}$, we define \begin{eqnarray} \label{def_tau_t}
  \lat{t} & := & \max\{\iat{k}: \iat{k} \leq t\},\\ \label{def_n_t}
\te{t}& :=& \max\{k: \iat{k} \leq t\}, \\ \label{def_delta_tau_t}
\dt{k}& := &\iat{\nk}-\iat{\ck}. \end{eqnarray}
 Additionally, $\gd(x,\mu,\Sigma)$ denotes
the density of a Gaussian random variable with mean $\mu\in\mR^\nd$
and covariance matrix $\Sigma\in \mR^{\nd\times\nd}$ evaluated at
$x\in\mR^\nd$. Finally, to avoid overloading the notation, any
constant that appears in all subsequent assumptions and proofs is
denoted by $C$.

\subsection{Setup} 

We consider a general setup, where our target process is a CTMP
evolving on state space $\texttt{X}$
with unknown transition density with respect to reference measure $\mbox{d}y$:
\begin{equation}\label{eqn:tr_dens}
\trd{x}{\y}{t} = \mathbb{P}(X_{s+t}\in\mbox{d}\y|X_s=x)/\mbox{d}\y~~\mbox{
  for $s,t>0$}.
\end{equation}
For notational simplicity we have assumed that the target process is
time-homogeneous and the initial state value is known
$X_0=x_0$. Extensions to cases
where the initial value is from a known distribution are
straightforward, while applicability to time-inhomogeneous processes would follow by extending the current setup
under appropriate smoothness conditions in the time variable.  

While the transition density is assumed to be intractable, under mild regularity conditions it
will satisfy Kolmogorov's forward equation
\begin{equation} \label{eq:PDE}
  \frac{\partial}{\partial t} \trd{x}{\y}{t} = \tro \trd{x}{\y}{t},
\end{equation}
where $\tro$ is the forward operator of the process, acting on
$\y$. For Markov jump processes this is just the rate matrix, while
for the diffusion (\ref{eq:SDE}) this is the second order differential
operator defined via
\begin{align*}
  \tro f(y) = -\sum_{i=1}^d \frac{\partial}{\partial y_i}
  [\bb_i(y)f(y)]+\sum_{i=1}^d\sum_{j=1}^d \frac{\partial}{\partial
    y_i\partial y_j} [\dm_{ij}(y)f(y)].
\end{align*}
To implement importance sampling we need a proposal process. We
consider a family of proposal processes, parameterised by
$\theta$, and assume these have a known transition density, denoted by
$\prd{x}{y}{t}{\theta}$. The proposal process will also have a family
of forward operators, which we denote by $\pro$. As an example,
when analysing diffusions, a natural choice of proposal processes will
be diffusions with constant drift and diffusion coefficient. In this
case $\theta$ would be the value of the drift and diffusion
coefficient, and the transition density would be Gaussian.  Throughout
the following we assume that for any $\theta$, any $t>0$ and any $x,$ 
the support of $\prd{x}{\cdot}{t}{\theta}$ contains the support of
$\trd{x}{\cdot}{t}$.

\subsection{The CIS algorithm}

We are now in a position to give a formal description of the algorithm outlined in
Section \ref{S:deriv}. The CIS algorithm simulates values of the process and importance
weights at times indicated by a renewal process. If $\iat{\ck},$ for
some $ k\in\kspace,$
is the most recent renewal time and $\x{\ck}, \w{\ck}$ denote
the value of the process and importance weight at that time, then the
algorithm iterates two steps. First, the waiting time for the next
renewal
$\iat{\nk}$ is simulated as the time to the event of a renewal process
with time-dependent intensity $\lambda(s), s>0$. If $\iat{\nk}<\th$ then
$\x{\nk}$ is generated from the proposal density
$\prd{\x{\ck}}{\cdot}{\dt{k}}{\theta_{\ck}}$ and the
weight is updated to
\[
\w{\nk}=\w{\ck}\mrw{\x{\ck}}{\x{\nk}}{\dt{k}}{\theta_{\ck}},
\]
where $\mrw{x}{y}{u}{\theta}$ is referred to as the
\textit{incremental weight} and is defined by
\begin{align}\label{eq:iweight}
 \mrw{x}{y}{u}{\theta}
=1+\frac{\lsb\tro-\prot{\theta}
\rsb \prd{x}{Y}{u}{\theta}\Big|_{Y=y}}{\lambda(u)\prd{x}{y}{u}{\theta}}.
\end{align}
If $\iat{\nk}>\th$ then the algorithm stops and outputs a distribution
for the value of the process at time $\th$, namely
$\prd{\x{\ck}}{\cdot}{\th-\iat{\ck}}{\theta_{\ck}}$, along with the
importance weight $\ww{\th}=\w{\ck}$. Full details are given in
Algorithm {\algcis} below.

\begin{description} 

\item[{\bf Algorithm \algcis}:] {\it Continuous-time Importance Sampling
    algorithm} (CIS).

\item[{\it Input}:] $\th$ a time to stop and $x_0$ the initial value.

\item[{\it Initialise}:] Set $k=\kstart, \iat{0}=0, 
\x{0}=x_0, \w{0}=1$ and choose $\theta_0$.
  
\item[{\it Iterate}:] Repeat
  \begin{enumerate}
    
  \item draw $u,$ the interarrival time of a renewal process
    of rate $\lambda(s)$, and set $\iat{\nk}=\iat{\ck} + u$,
    
  \item if $\iat{\nk} > \th$, set $\ww{\th}=\w{\ck}$ and stop.
  \item\label{item:cis_update} Otherwise

    \begin{enumerate}
    \item\label{item:cis_up_draw} draw $\x{\nk}\sim
    \prd{\x{\ck}}{\cdot}{\dt{k}}{\theta_{\ck}}$,
  \item\label{item:cis_up_wei} compute the incremental weight using \eqref{eq:iweight}
    \begin{align*}
      \iw{\nk} =
      \mrw{\x{\ck}}{\x{\nk}}{\dt{k}}{\theta_{\ck}}
    \end{align*}
    and update the weight $\w{\nk}=\w{\ck}\iw{\nk}$, 

  \item choose $\theta_{\nk} = \theta(\tro, \x{\nk})$ and set $k=k+1$.
    \end{enumerate}
  \end{enumerate}

\item[{\it Output}:] A distribution for the value of the process at
  time $\th$, namely \\ $\prd{\x{\ck}}{\cdot}{\th-\iat{\ck}}{\theta_{\ck}}$, and the
  corresponding importance sampling weight~$\ww{\th}$.

\end{description}

The output of the CIS algorithm should be interpreted as a (signed) measure valued
estimate of the transition measure, i.e. $\ww{\th}
\prd{\x{\ck}}{\cdot}{\th-\iat{\ck}}{\theta_{\ck}}$ is an estimate of $
\trd{x_0}{\cdot}{\th}$ and it can be  
 readily used to estimate
functional expectations, $\mE_p(f(X_{\th}))$, where $\mE_p$
denotes expectation with respect to the target process. Specifically,
if we repeat    Algorithm~{\algcis} $N$~times and denote the output of
the $i$th run to be weight $\ww{\th}^{(i)}$ and distribution
\begin{equation*}q^{(i)}(\cdot):=\prd{\x{T}^{(i)}}{\cdot}{\th-\lat{\th}^{(i)}}{\theta_{\te{\th}}^{(i)}},\end{equation*}
then our estimate of $\mE_p(f(X_{\th}))$  is
\begin{equation} \label{eq:unbias} 
\frac{1}{N} \sum_{i=1}^N \ww{\th}^{(i)}
\int f(\xx{\th}) q^{(i)}(\xx{\th}) \mbox{d}\xx{\th}.
\end{equation}
If we cannot calculate the integral in the above expression
analytically, we can approximate it by a standard Monte Carlo estimate
as we can draw samples from the proposal process density
$q^{(i)}(\cdot) $. In particular, we can simulate $\xx{\th}^{(i)}$ from
$q^{(i)}(\cdot)$, and approximate the expectation by $\sum_{i=1}^N
\ww{\th}^{(i)}f(\xx{\th}^{(i)})/N$. 
%
%

We should refer to Algorithm {\algcis} as generic CIS, and a
more detailed version will be considered in Section \ref{sec:CIS_diff}
to specifically deal with
SDEs. This will allow for precise and rigorous stability and validity
analysis in Section~\ref{sec:validity-diffusions}.
However, even at this level of generality we can identify convenient ingredients
for ensuring stability and unbiasedness, as
outlined in the next section.

\section{Unbiasedness and stability of generic
  CIS} \label{sec:validity-cis}

In this section we provide some general results about stability and
unbiasedness of the generic CIS algorithm. These are then used in the
particular case of SDEs in
Sections~\ref{sec:CIS_diff}~and~\ref{sec:validity-diffusions}, and may
be similarly applied in other settings of interest.

More specifically, recall that the output of Algorithm~$\algcis$ amounts to
the weight  $\ww{\th}$ and the distribution
$\prd{\x{\ck}}{\cdot}{\th-\iat{\ck}}{\theta_{\ck}}$, thus the
expectation of the
resulting density estimate is
%
\begin{equation} \label{eqn:est_exp}
\ttrd{x_0}{\y}{T} := \ecis\left[\WW{\th}
 \prd{\X{T}}{\y}{\th-\lat{\th}}{\theta_{\te{\th}}}\right],
\end{equation}
where $\ecis$ denotes the expectation with respect
to the probability space generated by the CIS algorithm until target time horizon $\th$, i.e. with
respect to the renewal process, the proposal process and its
parameters. Consequently, unbiasedness means
$\trd{x_0}{\y}{T} = \ttrd{x_0}{\y}{T}$ for every $T>0,$ and requires
suitable assumptions.

Furthermore, the tails of $\WW{\th}$, the random variable of the weight $\ww{\th}$ returned by the
CIS algorithm, determine the stability of the CIS based estimators,
e.g. \eqref{eq:unbias},
and in particular we need to establish that  $\WW{T} \in
L^1,$ otherwise~\eqref{eqn:est_exp} makes no sense. Thus stability of
the algorithm is studied by
examining conditions that guarantee that $\WW{T}$
is in $L^p.$  Stability and unbiasedness are studied in the following subsections.

\subsection{Stability of Weights}\label{subsec:weights}
To allow for precise statements let $\big(\mathcal{F}_t^{X,
    \lambda}\big)_{t\geq 0}$ be the filtration of the proposal
process $X$ together with the Poisson event process $\tau$, i.e. the
filtration generated by the CIS algorithm. The following two
assumptions imply a generic result regarding stability of weights.
\begin{assumption}[Uniform boundedness of moments of incremental
  weights]\label{assu:momentunif_bdd} 
For a constant \mbox{$p^* > 1$} consider the family of random variables 
\[
\mathcal{H}_k := \ecis\left[|H_k|^{p^*}  \Big| \mathcal{F}_{\tau_{k-1}}^{X,
    \lambda}
    \right],
\]
and assume that for some $p^*>1$ there exists a constant
  $C(p^*),$ s.t. for every $k=1,2,\dots$ 
  \begin{equation}
    \label{eq:momentunif_bdd}
    \mathcal{H}_k \leq C(p^*) \quad \textrm{a.s.}
  \end{equation}
\end{assumption}
\begin{remark}
  If the choice of $\theta_{k+1}$ in step 3(c) of Algorithm $\algcis$  depends
  on $x_{\tau_{k+1}}$ only and the intensity function of the Poisson
  process used in step~1 of Algorithm $\algcis$ is fixed, then
  condition \eqref{eq:momentunif_bdd} is implied by the more tractable
  \begin{equation}
    \label{eq:momentunif_bdd_smpl}
\sup_{x} \ecis\left[|H_1|^{p^*}  \Big| X_0=x
    \right] < C(p^*),
 \end{equation}
which will be verified for the diffusion case in Section \ref{sec:validity-diffusions}.
\end{remark}
\begin{assumption}[Finitness of the generating function of $\te{T}$]\label{assu:momentunif_bdd_gen_f}
Consider the sequence $\iat{\ck}$ of renewal times generated
by Algorithm {\algcis} and the resulting $\te{T}$ as defined in
equation \eqref{def_n_t}. Assume that the generating function of $\te{T}$
is everywhere finite, i.e. $\ecis [\exp\{c\te{T}\}] <
 \infty$ for every $c \in \mathbb{R}.$
\end{assumption}
The stability of weights is then guaranteed by the following result.
\begin{theorem}\label{th:tot_weight} Under Assumptions
  \ref{assu:momentunif_bdd} and \ref{assu:momentunif_bdd_gen_f} the weight
  $\WW{\th}$ returned by Algorithm {\algcis} is in $L^p$ for all $p<p^*$.
\end{theorem}

\begin{proof}[Proof of Theorem \ref{th:tot_weight}]
Fix arbitrary $p < p^*$ and let $\varepsilon:= p^*-p.$ Following Assumption \ref{assu:momentunif_bdd}, define \[ K:= C(p^*)^{\frac{p}{p^*}}  =
C(p+\varepsilon)^{\frac{p}{p + \varepsilon}} < \infty. \]
We first show that 
\begin{eqnarray}
\ecis |W_T|^p &  \leq & \sum_{j=0}^{\infty} K^{j}\big[\mathbb{P}(\te{T} = j)\big]^{\epsilon/p^*}. 
\label{eq:bdd_aux}
\end{eqnarray}
To this end we use the H\"older inequality (with ${\varepsilon \over
  p^*}+{p \over p^*}=1$) term by term in
the sum below as follows.
\begin{eqnarray*}
\ecis |W_T|^p & =  & \ecis \Big[\prod_{i=0}^{\te{T}} |H_i|^p\Big] \;\; = \;\; \ecis \Big[ \sum_{j=0}^{\infty} \mathbb{I}(\te{T} = j) \prod_{i=0}^j |H_i|^p  \Big] \\
& = &  \sum_{j=0}^{\infty} \ecis \Big[ \mathbb{I}(\te{T} = j) \prod_{i=0}^j |H_i|^p  \Big] \\ 
&\leq & \sum_{j=0}^{\infty} \big[\mP (\te{T} = j)\big]^{{\epsilon\over
    p^*}}  \bigg( \ecis \Big[ \prod_{i=0}^j |H_i|^{p^*}
\Big]\bigg)^{{p\over p^*}} \\
& =  & \sum_{j=0}^{\infty} \big[\mP (\te{T} = j)\big]^{{\epsilon\over
    p^*}}  \Bigg( \ecis \Bigg\{ \ecis \bigg[ \prod_{i=0}^j |H_i|^{p^*}  \bigg| \mathcal{F}^{X, \lambda}_{\tau_{j-1}} \bigg] \Bigg\} \Bigg)^{{p \over p^*}} \\
& = & \sum_{j=0}^{\infty} \big[\mP (\te{T} = j)\big]^{{\epsilon\over
    p^*}}  \Bigg( \ecis \bigg\{\prod_{i=0}^{j-1}
|H_i|^{p^*} \; \ecis \Big[ |H_j|^{p^*}  \Big|  \mathcal{F}^{X, \lambda}_{\tau_{j-1}} \Big] \bigg\} \Bigg)^{{p \over p^*}} \\
& \leq  & \sum_{j=0}^{\infty} \big[\mP (\te{T} =
j)\big]^{{\epsilon\over p^*}}  \Bigg(  C(p^*) \ecis
\Bigg\{\prod_{i=0}^{j-1} |H_i|^{p^*}  \Bigg\} \Bigg)^{{p \over
    p^*}} \leq  \dots \\
& \leq & \sum_{j=0}^{\infty}
C(p^*)^{jp/p^*}\big[\mP (n(t) =
j)\big]^{\epsilon/p^*} \; = \; \sum_{j=0}^{\infty}
K^{j}\big[\mP (\te{T} = j)\big]^{\epsilon/p^*},
\label{eq:}
\end{eqnarray*}

where the last conditioning is repeated recursively.

Now by Assumption \ref{assu:momentunif_bdd_gen_f}, the moment generating function of $\te{T}$ is finite for $$c:= {p^* \over \epsilon} \log ( 2K).$$
This implies that for some constant $M < \infty,$ $$\mP (\te{T} = j)
\; \leq \; M (2K)^{-{p^* \over \epsilon}j}$$   
and consequently by \eqref{eq:bdd_aux}, we obtain
\begin{eqnarray}
\ecis |W_T|^p &  \leq & \sum_{j=0}^{\infty} K^{j}\big[\mP (\te{T} = j)\big]^{\epsilon/p^*} \;\leq \; M ^{\epsilon/p^*} \sum_{j=0}^{\infty} K^{j}(2K)^{-j} \; < \infty. \nonumber 
\label{eq:}
\end{eqnarray}
Thus $W_T \in L^p.$
\end{proof}

\subsection{Unbiasedness}\label{subsec:unbiasedness} 
The generic CIS relies on the Kolmogorov forward equation, hence the fundamental
assumption that underpins the approach is:
\begin{assumption}[Validity of forward
  equations]~ \label{aasu:Komogorovl_unique} 
\begin{enumerate}[label=(\alph*)]
\item The transition density of the target process
  $\trd{x}{\y}{t},$ defined in~\eqref{eqn:tr_dens}, is the unique solution
  of the Kolmogorov forward equation~\eqref{eq:PDE}.
\item Similarly, for every $\theta,$ the transition density
  of the proposal process $\prd{x}{y}{t}{\theta},$ is the unique solution
  of the Kolmogorov forward equation
\begin{equation} \label{eq:forward_proposal}
  \frac{\partial}{\partial t}\prd{x}{y}{t}{\theta} = 
\pro \prd{x}{y}{t}{\theta}.
\end{equation}
\end{enumerate}
\end{assumption}
Unbiasedness of the CIS algorithm will follow from 
verifying that $\ttrd{x_0}{\y}{T}$ defined by~\eqref{eqn:est_exp}
also satisfies the Kolmogorov forward equation~\eqref{eq:PDE}, and
hence by uniqueness is the same as  $\trd{x}{\y}{T}.$

To this end we relate the CIS algorithm to Piecewise Deterministic
Markov Processes \cite{MR790622} (see also \cite{MR1283589,MR1680267}) . These are processes that have
stochastic jumps at event times of a point process, but where the
state evolves deterministically between event times. 

We can restrict our attention to the case where the statespace of a Piecewise Deterministic Markov Process (PDP) can be written as a
single open set (rather than a sum of indexed open sets,
c.f. \cite{MR790622}), which simplifies notation. The following
objects define a PDP:
\begin{itemize}
\item The
state space $(E, \mathcal{B}(E))$  of such a PDP process  $Z_t \in E$; 
\item $\lambda^{\textrm{PDP}}(z): E \to \mathbb{R}_{+},$ a state dependent rate function of the event
  times;
\item $\psi: E \to \mathbb{R}^l,$ a vector field that specifies the deterministic
  dynamics
  between jump times. If the vector field has a unique integral curve solution
  $\Psi(z,t): E\times \mathbb{R}_+ \to E$ and  there are no events in $[t,t+s],$ then $Z_{t+s} = \Psi(Z_t,  s).$ 
\item $Q(z, \cdot): E \times \mathcal{B}(E) \to
  [0,1],$  a transition kernel for the process at jump times, i.e. if
  there is an event at $t,$ then $Z_t \sim Q(z_{t^-}, \cdot).$
\end{itemize}

In the following definition we write the process generated
by the CIS (i.e. Algorithm~{\algcis}) as a PDP.
\begin{defi}[CIS as PDP]\label{CIS_PDP}~
\begin{itemize}
\item Let $z = (u,x,w)  \in \mathbb{R}_{+}\times \texttt{X} \times
  \mathbb{R} =: E,$ where $\texttt{X}$ is the state space of the
  target CTMP. Here $u$ is the time since the most
  recent jump (or the current time if there has been no jump); $x$ is
  the value of the process simulated at the last event time (or the
  initial value of the process if there has been no jump), and $w$ is
  the current value of the weight, that is 
\[ z_t := (t-\lat{t}, x_{\lat{t}}, w_{\lat{t}}).\]
\item Take $\lambda(u),$
  the rate of the CIS algorithm interarrival renewal process, and let \[\lambda^{\textrm{PDP}}(z) := \lambda(u).\]
\item Define the deterministic dynamics between events as \[ \Psi\big((u,x,w), s
  \big) := (u+s, x, w),\]
or equivalently 
\[\psi(z_t):= \frac{\partial z_t}{\partial t} = (1,0,0).
\]
\item Given $z_{t^{-}}=(u,x,w),$ the transition kernel  $Q(z_{t^{-}}, \cdot)$ is defined
  through sampling \[ y \sim
    \prd{x}{\cdot}{u}{\theta(x)}, \]
and setting 
\[z_t:= \big(0, y, w  \mrw{x}{y}{u}{\theta(x)}\big).
\]
\end{itemize}
Algorithm {\algcis} generates the process $Z_t$ defined as above with
initial value $Z_0 = (0, x_0, 1).$ 
\end{defi}

Denote by $M(E)$ the family of real valued measurable functions on $E,$
and by $M_b(E) \subset M(E),$ the subfamily of all bounded
functions. Let $\|g\|_{\infty}$ be the sup norm of $g \in M(E).$ Let
$\mathcal{T}(s)$ be the semigroup on $M_b(E)$ associated with the
transition kernels of our PDP $Z_t$ and let $\mathcal{A}$ be its full generator.

In the CIS setting, where in particular there is no active boundary
of $E$ (i.e. the set of boundary points of $E$ that can be reached
from $E$ via integral curves is empty), the full generator $\mathcal{A}$
and its domain $\mathcal{D}(\mathcal{A})$ are characterised by the
following result, which is a straightforward simplification of Theorem
11.2.2 of \cite{MR1680267} and the preceding discussion of the Dynkin formula.
\begin{theorem}\label{th:PDP_gen} Let $Z_t$ be the PDP constructed in
  Definition \ref{CIS_PDP} and let $g \in M(E)$ satisfy the following
  conditions:
\begin{enumerate}
\item for each $z \in E$ the function $h(t) := g\big(\Psi(z,t)\big)$ is
  absolutely continuous on $(0, \infty)$;
\item for each $t\geq 0$,
\[
\mathbb{E}\left( \sum_{i=1}^{n_t}\left| g(Z_{\tau_i}) -
    g(Z_{\tau_i^{-}})\right|\right) < \infty.
\]
\end{enumerate} 
Then $g \in  \mathcal{D}(\mathcal{A}).$ Moreover, 
 at $z=(u,x,w),$
\begin{eqnarray} \nonumber
\mathcal{A}g(z) & =  & \psi(z) \cdot \nabla g(z) + \lambda^{PDP}(z)\int_{E}
\big(g(z') - g(z)\big) Q(z, \mbox{d} z') \\ \label{generator_form}
&=& \frac{\partial}{\partial u}g(u,x,w) + \lambda(u)\int_{\texttt{X}}
    \Big(g\big(0, y, w  \mrw{x}{y}{u}{\theta(x)}\big) - g(u,x,w)\Big) \prd{x}{y}{u}{\theta(x)}\mbox{d}y.
\end{eqnarray}
Furthermore, the Dynkin formula holds for $g,$ hence if $Z_0 =
z_0$, then:
\begin{eqnarray}\label{Dynkin}
\mathbb{E}\big[ g(Z_T)\big] & = & g(z_0) + 
\mathbb{E}\left[ \int_0^T \mathcal{A}g(Z_s)\mbox{d}s \right].
\end{eqnarray}
\end{theorem}

To prove unbiasedness we will relate the generator $\mathcal{A}$ of
the PDP to the forward operator $\mathcal{K}$ of the target for a
suitable class of functions. This will require some further
assumptions. Denote by $\mathcal{L}$ the generator of
the target process (i.e. the adjoint of $\mathcal{K}$) and its domain
by $\mathcal{D}(\mathcal{L})$.

\begin{assumption}[Regularity conditions] \label{aasu:test_f}
  Assume there exists a family of functions $\mathcal{S} \subseteq
  \mathcal{D}(\mathcal{L})\cap M_b(\texttt{X})$ satisfying the following conditions:
\begin{enumerate}[label=(\alph*)]
\item For each $f \in \mathcal{S}$ the function $g_f: E \to
  \mathbb{R}$ defined as 
\begin{equation} \label{def_gf} g_f(u,x,w) := w\int_{\texttt{X}}f(y)
  \prd{x}{y}{u}{\theta(x)}\mbox{d}y, \end{equation}
satisfies conditions 1. and 2. of Theorem \ref{th:PDP_gen} and hence
$g_f \in \mathcal{D}(\mathcal{A}).$
\item For each $f \in \mathcal{S},$ $x \in \texttt{X}$ and $u \in
  (0,T),$ time differentiation can be moved under the integral sign, i.e.
\begin{eqnarray} \label{deriv:q}
\frac{\partial}{\partial u}\left\{\int_{\texttt{X}}f(y)
  \prd{x}{y}{u}{\theta(x)}\mbox{d}y\right\} & = &
\int_{\texttt{X}}f(y)  \frac{\partial}{\partial u}\left\{
  \prd{x}{y}{u}{\theta(x)}\right\}
\mbox{d}y; \\ \label{deriv:p}
\frac{\partial}{\partial u}\left\{\int_{\texttt{X}}f(y)
  \ttrd{x}{y}{u}\mbox{d}y\right\} & = &
\int_{\texttt{X}}f(y)  \frac{\partial}{\partial u}\left\{
  \ttrd{x}{y}{u}\right\}
\mbox{d}y.
\end{eqnarray}
\item For each $f \in \mathcal{S},$ $x \in \texttt{X}$ and $t>0,$ the
  following adjoint equations
  hold:
\begin{eqnarray} \label{adjoint:q}
\int_{\texttt{X}} \prd{x}{y}{t}{\theta(x)}\mathcal{L}f(y) \mbox{d}y & = &
\int_{\texttt{X}} f(y) \mathcal{K}\prd{x}{y}{t}{\theta(x)}\mbox{d}y;
\\ \label{adjoint:p}
\int_{\texttt{X}} \ttrd{x}{y}{t}\mathcal{L}f(y) \mbox{d}y & = &
\int_{\texttt{X}} f(y) \mathcal{K}\ttrd{x}{y}{t}\mbox{d}y.
\end{eqnarray}
\item $\mathcal{S}$ is separating, i.e:
\begin{eqnarray}\nonumber
\left\{  \forall_{f\in \mathcal{S}} \quad 
\int_{\texttt{X}} f(y) h_1(y) \mbox{d}y =
\int_{\texttt{X}} f(y) h_2(y) \mbox{d}y 
  \right\} & \implies  & h_1 = h_2 \quad \textrm{a.s.}
\end{eqnarray}
\end{enumerate}
\end{assumption}
The following lemma helps to establish part $(a)$ of
Assumption~\ref{aasu:test_f}, and does not rely on further details of
the setting.
\begin{lemma} If Assumptions~\ref{assu:momentunif_bdd}
  and~\ref{assu:momentunif_bdd_gen_f} hold, then for any $f \in
  \mathcal{S}$ the function  $g_f$ defined in~\eqref{def_gf} satisfies Condition 2 of Theorem~\ref{th:PDP_gen}.
\end{lemma}
\begin{proof}
Recall that  $\mathcal{S} \subseteq
  \mathcal{D}(\mathcal{L})\cap M_b(\texttt{X}).$ Fix $f \in
  \mathcal{S}$ and let $C:=\max\{1, \|f\|\}.$ Recall the
definition of $k-$th incremental weight $H_k$ and define
$H_k^*:=C(|H_k|+1) \geq 1.$ Let $W_{\tau_{i}}^*$ and $\WW{\th}^*$ be the
analogues of $W_{\tau_{i}}$ and $\WW{\th},$ respectively, 
built from $H_k^*$ instead of $H_k.$   Note that if $\{H_k\}$ satisfies
Assumption~\ref{assu:momentunif_bdd}, so does  $\{H_k^*\}.$ Hence
Theorem \ref{th:tot_weight} yields $W_{t}^* \in L^p$
for all $p< p^*$ and $t\leq T.$ Now
compute
\begin{align*}
\mathbb{E}&\left( \sum_{i=1}^{n_t}\left| g_f(Z_{\tau_i}) -
    g_f(Z_{\tau_i^{-}})\right|\right) \; = \; 
\mathbb{E}\left( \sum_{i=1}^{n_t}\left| g_f(0, X_{\tau_i}, W_{\tau_i}) -
    g_f(\tau_i - \tau_{i-1}, X_{\tau_{i-1}},
                                            W_{\tau_{i-1}})\right|\right)
  \\
& = \; \mathbb{E}\left( \sum_{i=1}^{n_t}\left|
W_{\tau_{i-1}}\left(H_i \int_{\texttt{X}}f(y)
  \prd{X_{\tau_{i}}}{y}{0}{\theta_{i}}\mbox{d}y - \int_{\texttt{X}}f(y)
  \prd{X_{\tau_{i-1}}}{y}{\tau_i - \tau_{i-1}}{\theta_{i-1}}\mbox{d}y
\right) \right|\right) \\
& \leq \; 
\mathbb{E}\left( \sum_{i=1}^{n_t}\left|
W_{\tau_{i-1}}\right| \left(\left| H_i \right| C + C\right) \right)
\;  \leq \;  \mathbb{E}\left( \sum_{i=1}^{n_t}\left|
W_{\tau_{i}}^*\right|  \right) \; \leq \;  
\mathbb{E}\left( \sum_{i=1}^{n_t}\left|
W_{n_t}^*\right|  \right)
\\ & = \; \mathbb{E}\left( n_t \left|
W_{t}^*\right|  \right) \; \leq \;  \left( \mathbb{E}\left(
                                    n_t^{\frac{p}{p-1}}\right) \right)^{\frac{p-1}{p}} \left( \mathbb{E}\left(
                                    \left|
W_{t}^*\right|^{p}\right) \right)^{\frac{1}{p}} \; <\; \infty,
\end{align*}
where $1<p<p^*$ and we used that $n_t$ has all moments by Assumption~\ref{assu:momentunif_bdd_gen_f}.
\end{proof}
A key step to verifying unbiasedness of CIS is to compute how the generator
acts on functions of the form \eqref{def_gf}.
 \begin{lemma}\label{lem_Af}
Let Assumptions \ref{aasu:Komogorovl_unique} and \ref{aasu:test_f} be
satisfied and let $g_f$ be given by \eqref{def_gf}. Then 
\begin{eqnarray}\label{eq_Af}
\mathcal{A}g_f(u,x,w) & = & w\int_{\texttt{X}}f(y)
  \mathcal{K}\prd{x}{y}{u}{\theta(x)}\mbox{d}y.
\end{eqnarray}
\end{lemma}
\begin{proof}
Since Assumption \ref{aasu:test_f} $(a)$ is satisfied, using
expressions \eqref{generator_form} and \eqref{eq:iweight}, we obtain
\begin{eqnarray} \nonumber
\mathcal{A}g_f(u,x,w) 
& = & w \frac{\partial}{\partial u}\left\{\int_{\texttt{X}}f(y)  \prd{x}{y}{u}{\theta(x)}\mbox{d}y\right\}
 \\ \nonumber
& & + \; \lambda(u) w \int_{\texttt{X}}\left[ 
 \mrw{x}{y}{u}{\theta(x)}  f(y) -  \int_{\texttt{X}}f(y)  \prd{x}{y}{u}{\theta(x)}\mbox{d}y
\right] \prd{x}{y}{u}{\theta(x)}\mbox{d}y\\ \nonumber
& \stackrel{\textrm{\eqref{eq:iweight}}}{=} & w \frac{\partial}{\partial u}\left\{\int_{\texttt{X}}f(y)  \prd{x}{y}{u}{\theta(x)}\mbox{d}y\right\}
 \\ \nonumber
& & + \; \lambda(u) w \int_{\texttt{X}}\left[ 
f(y) -  \int_{\texttt{X}}f(y)  \prd{x}{y}{u}{\theta(x)}\mbox{d}y
\right] \prd{x}{y}{u}{\theta(x)}\mbox{d}y \\ \label{expr:Af}
& & + \; w \int_{\texttt{X}}
f(y) \lsb\tro-\prot{\theta(x)}
\rsb \prd{x}{y}{u}{\theta(x)} \mbox{d}y \; = \; A + 0 + B.
\end{eqnarray} 
By Assumptions \ref{aasu:test_f} $(b)$ and
\ref{aasu:Komogorovl_unique} $(b)$
\begin{eqnarray*}
A = w\int_{\texttt{X}}f(y)
  \mathcal{K}_{\theta(x)}\prd{x}{y}{u}{\theta(x)}\mbox{d}y,
\end{eqnarray*}
which combined with \eqref{expr:Af} yields \eqref{eq_Af}.
\end{proof}

\begin{theorem}\label{thm:generic_unbiased} Under Assumptions
  \ref{assu:momentunif_bdd}, \ref{assu:momentunif_bdd_gen_f},
  \ref{aasu:Komogorovl_unique} and \ref{aasu:test_f}, the CIS Algorithm
  {\algcis} is unbiased, i.e. 
  \begin{eqnarray}\label{eqn:generic_unbiased}
\trd{x_0}{\y}{T} & = & \ttrd{x_0}{\y}{T} \qquad \textrm{for every}
\quad T>0,
\end{eqnarray}
where $\ttrd{x_0}{\y}{T}$ is defined in equation \eqref{eqn:est_exp}.
\end{theorem}
\begin{proof}
Let $f \in \mathcal{S}.$ Since $f$ is bounded and by Theorem
\ref{th:tot_weight}  $\WW{\th}\in L^p$ for some $p>1,$ the RHS of
\eqref{G-U-1} below, exists by H\"older inequality, and by Fubini-Tonelli also the LHS of
\eqref{G-U-1} exists. Furthermore, by Assumption \ref{aasu:test_f}, Dynkin
formula \eqref{Dynkin} holds and 
writing CIS as PDP (c.f. Definition \ref{CIS_PDP}), yields
\begin{align} \label{G-U-1}
 \int_{\texttt{X}}f(y) \ttrd{x_0}{\y}{T} \mbox{d} y & = \ecis\left[ f(y) \WW{\th}
 \prd{\X{T}}{\y}{\th-\lat{\th}}{\theta_{\te{\th}}}\right]
 \\ \nonumber 
& =
\mathbb{E}\big[ g_f(Z_T)\big]  =  g_f(z_0) + 
\mathbb{E}\left[ \int_0^T \mathcal{A}g_f(Z_s)\mbox{d}s \right]
\\ \nonumber & =   f(x_0) + 
\int_0^T \mathbb{E} \left[ \mathcal{A}g_f(Z_s) \right] \mbox{d}s.
\end{align}
Now use Lemma \ref{lem_Af}, then Assumption \ref{aasu:test_f} $(c)$,
then Fubini-Tonelli, then Assumption \ref{aasu:test_f} $(c)$ again, to compute
\begin{align} \nonumber 
 \int_{\texttt{X}}f(y) \ttrd{x_0}{\y}{T} \mbox{d} y & 
= f(x_0) +
\int_0^T \mathbb{E} \left[ \WW{s} \int_{\texttt{X}}f(y)
  \mathcal{K}\prd{X_{\tau_s}}{y}{s-\tau_s}{\theta_{\te{s}}}\mbox{d}y           \right]        \mbox{d}s \\
  \nonumber
& 
= f(x_0) +
\int_0^T \mathbb{E} \left[ \WW{s} \int_{\texttt{X}}
  \prd{X_{\tau_s}}{y}{s-\tau_s}{\theta_{\te{s}}} \mathcal{L} f(y)\mbox{d}y           \right]        \mbox{d}s \\
  \nonumber
& = f(x_0) +
\int_0^T \left[
\int_{\texttt{X}}\ttrd{x_0}{\y}{s} \mathcal{L}f(y) \mbox{d} y   \right]
  \mbox{d}s \\
  \nonumber
& = f(x_0) +
\int_0^T \left[
\int_{\texttt{X}}f(y) \mathcal{K}\ttrd{x_0}{\y}{s} \mbox{d} y   \right]        \mbox{d}s.
\end{align}
Differentiate the above equation in $T,$ use Assumption
\ref{aasu:test_f} $(b)$ to interchange differentiation and integration of
the LHS, and then Assumption \ref{aasu:test_f} $(d)$ to conclude that
\begin{equation*}
\frac{\partial}{\partial t}\ttrd{x_0}{\y}{t} = \mathcal{K} \ttrd{x_0}{\y}{t},
\end{equation*}
which by Assumption \ref{aasu:Komogorovl_unique} implies \eqref{eqn:generic_unbiased}.
\end{proof}

\section{The copycat CIS for diffusions} \label{sec:CIS_diff}

We now shift our attention to the case where the target CTMP is the
diffusion process \eqref{eq:SDE}, and present in detail how to choose
the renewal intensity and proposal density for a CIS implementation in
this context.

If $\pin>0$ and $0<\pex<1$ are two user-defined constants then the
renewal intensity is chosen as \begin{equation} \label{eq:lambda}
\lambda(s)=\pin
s^{\pex-1}. \end{equation} Simulating renewal times from such a process is
straightforward using the inverse transform. 
As we shall see
later, this choice of $\lambda(s)$ ensures that the incremental weight
is well behaved.

The proposal density is constructed by `adapting' it to the current
state $(\iat{\ck}, \x{\ck}, \w{\ck})$ of the CIS algorithm. Specifically,
the process is propagated at the next renewal time $\iat{\nk}$ according
to
\begin{align*}
  \d \XX{s} = \bb(\x{\ck})\d s + \dc(\x{\ck})\d\B_s,~~~~X_0=\x{\ck},\, s\geq 0.
\end{align*}
Notice that the drift and diffusion coefficient are constant and fixed
at their value given $\x{\ck}$, the value of the diffusion at the most recent renewal time. Hence, the proposal process is Gaussian
with tractable transition density parameterised by
$\theta_{\ck}=(\bb(\x{\ck}), \dc(\x{\ck}))$ and
$\prdcc{\ck}{x}{y}{t} = \gd\{y; x+t\bb(\x{\ck}), t\dm(\x{\ck})\}$.

Hereafter, the above settings are assumed to be the default for a CIS
implementation for diffusions, unless otherwise stated. We refer to
the resulting algorithm as copycat CIS since the proposal process
adapts at every renewal event, and denote it simply by CIS. The
algorithm is presented in Algorithm~{\algcc}.

\begin{description}

\item[{\bf Algorithm \algcc}:] {\it Copycat CIS for diffusions}.

\item[{\it Input}:] $\th$ a time to stop and $x_0$ the initial value.

\item[{\it Initialise}:] Set $k=\kstart, \iat{0}=0, \iw{0}=1, \x{0}=x_0,
  \w{0}=1$ and $\theta_0=(\bb(\x{0}), \dc(\x{0})))$.
  
\item[{\it Iterate}:] Repeat
  \begin{enumerate}
    
  \item draw $U\sim\Un(0,1)$, and set $\iat{\nk}=\iat{\ck} + \{-\pex\log(U)/\pin\}^{1/\pex}$,
    
  \item if $\iat{\nk} > \th$, set $\ww{\th}=\w{\ck}$ and stop.
  \item\label{item:cis_update} Otherwise

    \begin{enumerate}
    \item\label{item:cis_up_draw} draw $\x{\nk}\sim
    \prd{\x{\ck}}{\cdot}{\dt{k}}{\theta_{\ck}}$,
  \item compute the incremental weight 
    \begin{align*}
      \iw{\nk} =
      \mrw{\x{\ck}}{\x{\nk}}{\dt{k}}{\theta_{\ck}}\,\,\mbox{as
        in \eqref{eq:iweightcc}}
    \end{align*}
    and update the weight $\w{\nk}=\w{\ck}\iw{\nk}$, 

  \item set $\theta_{\nk}=(\bb(\x{\nk}), \dc(\x{\nk}))$ and $k=k+1$.
    \end{enumerate}
  \end{enumerate}

\item[{\it Output}:] A distribution for the value of the process at
  time $\th$, namely \\ $\prd{\x{\ck}}{\cdot}{\th-\iat{\ck}}{\theta_{\ck}}$, and a
  corresponding importance sampling weight $\ww{\th}$.

\end{description}

The formula for the incremental weight based on the above settings is
derived in the Lemma below, which is proved in the Appendix.

\begin{lemma}\label{lemma:iweightcc}
  The incremental weight $\mrwcc{}{x}{y}{u}$ for the copycat CIS
  presented in Algorithm {\algcc} is given by
  
\begin{align}\label{eq:iweightcc}
  \mrwcc{}{x}{y}{u}= & 1 +
  \frac{1}{\lambda(u)}\bigg(\frac{1}{2}\Big\{[\dm(y)-\dm(x)]:\qdderf{\theta}{x}{y}{u} +
  \dmdder(y):\one_{\nd\times\nd}
  \Big\}\notag\\
&+\lsb\dmder(y) \one_{\nd\times 1} - \bb(y) +
\bb(x)\rsb\cdot\qderf{\theta}{x}{y}{u} - \bder(y)\cdot\one_{\nd\times 1}\bigg),
\end{align}
where $\bder, \dmder, \dmdder$ as in Section \ref{section:notation} and
\begin{align}
  \label{eq:qders}
  \qderf{\theta}{x}{y}{u}&:=\frac{\nabla_y \prdcc{}{x}{y}{u}}{\prdcc{}{x}{y}{u}} =
  -\frac{\dm^{-1}(x)}{u}[y-x-u\bb(x)],\\
\qdderf{\theta}{x}{y}{u}&:=\frac{H_y \prdcc{}{x}{y}{u}}{\prdcc{}{x}{y}{u}} = \qderf{}{x}{y}{u}\qderft{}{x}{y}{u} - \frac{\dm^{-1}(x)}{u}.
\end{align}
\end{lemma}
The expressions above are presented conveniently in a matrix form
which is very useful from a programming and computational point of
view. In the univariate case ($\nd=1$), the formula for the
incremental weight simplifies to
\begin{align} \label{eq:CPW}
  \mrwcc{}{x}{y}{u}=&1 + \frac{1}{\lambda(u)}\Bigg(
  \frac{1}{\dm(x)u}\Bigg\{ \frac{\dm(y)-\dm(x)}{2}\lsb
  \frac{(y-x-u\bb(x))^2}{\dm(x)u} - 1 \rsb \\
  & \nonumber + [\bb(y)-\bb(x)-\dm^{'}(y)][ y-x-u\bb(x)]\Bigg\} +
  \frac{\dm^{''}(y)}{2} - \bb^{'}(y)\Bigg).
\end{align}
%


\section{Unbiasedness and stability of the copycat CIS for diffusions} \label{sec:validity-diffusions}

In this section we deal with stability and unbiasedness of the copycat
CIS for diffusions detailed in Algorithm {\algcc} in Section \ref{sec:CIS_diff}.
The settings for the copycat CIS, in particular the choice of the
renewal intensity function along with assumptions on drift and
diffusion coefficients will guarantee that the incremental weights
$\{\iw{k}\}_k$ are uniformly bounded, thus ensuring the stability of
the overall weight returned by the algorithm. We start with a
cautionary example followed by results on stability of weights and unbiasedness.

\subsection{A cautionary example} 

To understand the issues, consider the single incremental weight for the copycat scheme for a univariate diffusion (\ref{eq:CPW}).
First condition on $x$ and $u$, and consider this weight as a function of the new value of the process $Y$.
We can write $Y=x+u\bb(x)+\sqrt{u\dm(x)}Z$ where $Z$ is a standard normal random variable. Then we get
\begin{align*} 
  \mrwcc{}{x}{Y}{u}=&1 + \frac{1}{\lambda(u)}\Bigg(
  \frac{1}{\dm(x)u}\Bigg\{ \frac{\dm(Y)-\dm(x)}{2}\lsb
  Z^2 - 1 \rsb \\
  &+ [\bb(Y)-\bb(x)-\dm^{'}(Y)][ \sqrt{u\dm(x)}Z]\Bigg\} +
  \frac{\dm^{''}(Y)}{2} - \bb^{'}(Y)\Bigg).
\end{align*}
The term that causes problems with stability of the incremental weight is
\begin{equation}
 \frac{1}{\lambda(u)}
  \frac{1}{\dm(x)u}\Bigg\{ \frac{\dm(Y)-\dm(x)}{2}\lsb
  Z^2 - 1 \rsb \Bigg\},
\label{eq:incw}
\end{equation}
due to the $u$ in the denominator. Thus any instability in the weight is caused by small values of $u$, ie when
two renewal events occur very close to each other.

We have to be careful that the variability of this contribution to the incremental
weight does not blow-up too quickly as $u\rightarrow 0$. However, we immediately
see that if $\gamma $ is for example a Lipschitz continuous function and $Y$ is ``close to'' $x$, then the
term in the numerator of the curly bracket of (\ref{eq:incw}) has the effect of slowing
down this explosion. In fact, we can see  that
\[
 \dm(Y)-\dm(x)=\dm(x+u\bb(x)+\sqrt{u\dm(x)}Z)-\dm(x)=O_p(\sqrt{u}).
\]
Now when considering whether the incremental weight is in $L^1$, we are required to integrate with respect to the renewal distribution density
which includes a term $\lambda (u)$ which cancels with the corresponding term in the denominator of  (\ref{eq:incw})
and thus the integrand  will be $O_p({u^{-1/2}})$ for small $u$ and hence is integrable. This is the crux of the integrability
stability provided by the copycat scheme.

We note that, were we to be using a different scheme from the copycat one, where the proposal process variance was not
chosen to be close to $\gamma (x)$, then the integrand for the $L^1$ integrability will typically be $O_p({u^{-1}})$ and the
incremental weight will thus fail to be in $L^1$. Therefore $L^1$ stability of the weights will usually fail if we do not adopt
the copycat scheme.

Now we return to the copycat scheme and informally consider integrability in $L^p$ for some $p>1$.  Note that the
choice of $\alpha $ made no difference when considering $L^1$ stability as the term involving $\lambda $ cancelled.
However for $p>1$ this cancellation does not occur. By choosing $\lambda(u)\propto u^{\alpha-1}$ for some $\alpha<1$
the integrand for $L^p$ stability resembles
\[
{\lambda (u) (\gamma (Y) - \gamma (x) )^p \over u^p \lambda (u)^p} = O_p( u^{-p/2+(1-\alpha )(p-1)}  )
\]
for small $u$. The integrand is thus integrable for fixed $p$ whenever
\[
\alpha < {1 \over 2(1-p^{-1} ) }
\]
and in particular this holds whenever $\alpha \le 1/2$. In conclusion,
for the simplified setting of this informal calculation, $\alpha \le 1/2$ is sufficient to ensure finiteness of all polynomial moments of the incremental weights.

It is natural to ask whether the strategy of setting  $\lambda(u)\propto u^{\alpha-1}$ for some $\alpha<1$ can save non-copycat schemes where the proposal process variance was not chosen to be close to $\gamma (x)$.  However a similar calculation yields that the $L^p$ integrand has order
$u^{(1-\alpha ) ( p-1) -p}$ which is not integrable at $0$ for any choice of $0\le \alpha <1$, $p\ge 1$.

These informal calculations for the single incremental weight in the
univariate Lipshitz case demonstrate that the copycat scheme is more or less the only possible scheme to ensure stability of the incremental weights,
and this explains why we concentrate on this case from now on.
Although the arguments for the rest of this section are quite
technical, generalisation of the above informal calculations form their
core.

\subsection{Assumptions}

The following regularity conditions will be used to establish the
stability of weights and unbiasedness. They are stated in a rather
weak form with respect to the geometry induced by the diffusion
coefficient $\sigma,$ c.f. the remarks below about simplifying and
verifying them.

\begin{assumption}[Moments]\label{assu_for_CIS_moments}~
  \begin{itemize}
    \item[(i)] 
Polynomial growth: there exists $C < \infty,$ and $0< \kappa_{\bb} < \infty,$ s.t. for all $x,y \in
      \mathbb{R},$   
\begin{equation}
          \nonumber 
          \|\dc^{-1}(x)(\bb(y) - \bb(x))\| \leq C\big( 1+  \| \dc^{-1}(x)(y-x)
          \|^{\kappa_{\bb}}\big).          
        \end{equation}
\item[(ii)] $ \|\dc^{-1}(x)\bb(x)\| \leq C, $ for
  some $C< \infty$ and all $x \in \mR.$
    \item[(iii)] The derivatives $\partial \bb_i(x) / \partial x_i $
      exist and are bounded above, i.e. $\| \bder(x)\|  \leq C$
    for all $x \in \mathbb{R}^d$ and some $C < \infty.$ 
    \item[(iv)] The infinitesimal variance $\dm(x)
      = \dc(x) \dc^T(x)$ is such that the
      derivatives \begin{equation} \nonumber  \partial
      \dm_{ij}(x) / \partial x_j \qquad \textrm{and} \qquad \partial^2
      \dm_{ij}(x) / \partial x_i \partial x_j \end{equation} exist and are
      bounded above, i.e. for all $x \in \mathbb{R}^d$ and
      some $C < \infty,$ 
\begin{equation} \nonumber  \|\dmder(x)\| \leq C
        \qquad \textrm{and} \qquad \|\dmdder(x)\| \leq
        C.  \end{equation} 
    \item[(v)] The infinitesimal variance is uniformly bounded from
      below, i.e. there exists $c > 0,$ s.t. for all $x, z \in
      \mathbb{R}^d,$
      \begin{equation}
         \nonumber 
         z \cdot \dm(x) z  \geq c z \cdot z.
      \end{equation}

      \item[(vi)] The infinitesimal variance is locally spectrally H\"older
        continuous of order  $0< \kappa_{\dm} \leq 1$ with polynomial
        growth, in the following
        sense. There exists $C < \infty$ and $\kappa_{\dm} < m <
        \infty,$ s.t. for all $x,y \in
      \mathbb{R}^d,$   
        \begin{equation}
          \nonumber 
          \|\dc^{-1}(x)(\dm(y) - \dm(x))\| \leq C \big( \|\dc^{-1}(x)
          (y-x) \|^{\kappa_{\dm}}+ \|\dc^{-1}(x) (y-x) \|^{m} \big).
        \end{equation}
     \item[(vii)]  The diffusion coefficient is locally spectrally H\"older
        continuous of order  $0< \kappa_{\dc} \leq 1$ with polynomial
        growth, in the following
        sense. There exists $C < \infty$ and $\kappa_{\dc} < m <
        \infty,$ s.t. for all $x,y \in
      \mathbb{R}^d,$   
        \begin{equation}
          \nonumber 
          \|\dc^{-1}(x)(\dc(y) - \dc(x))\| \leq C  \big( \|\dc^{-1}(x)
          (y-x) \|^{\kappa_{\dc}}+ \|\dc^{-1}(x) (y-x) \|^{m} \big). \end{equation}
  \end{itemize}
\end{assumption}

\begin{remark} \label{remark_ver_assu} As for verifying the above
  conditions, note that $(i)$ is implied by  $(v)$ and any of 
  \begin{itemize}
  \item[$(i')$] The drift is bounded above, i.e. $\|\bb(x)\| \leq C,$
    for all $x \in \mathbb{R}^d$ and some $C < \infty,$
  \item[$(i'')$]  $  \|\bb(y) - \bb(x)\|  \leq  C \big( 1+  \| \dc^{-1}(x)(y-x)
          \|^{\kappa_{\bb}}\big) \; $ for all $x \in \mathbb{R}^d,$
          some $C < \infty.$
  \end{itemize} 
  Similarly, if $(v) $ holds, then $(vi)$ is implied by any of 
  \begin{itemize}
  \item[$(vi')$] $ \|\dm(y) - \dm(x)\|  \leq  C \big( \|\dc^{-1}(x)
     (y-x) \|^{\kappa_{\dm}} + \|\dc^{-1}(x)
     (y-x) \|^{m} \big), $  for all $x \in \mathbb{R}^d$ and some $C < \infty,$
\item[$(vi'')$]  The infinitesimal variance is H\"older continous and
  bounded above, i.e. $ \|\dm(y) - \dm(x)\|  \leq  C \|
     y-x \|^{\kappa_{\dm}}$ and  $\|\dm(x)\| \leq C$
    for all $x \in \mathbb{R}^d$ and some $C < \infty.$ 
  \end{itemize}
Morover, if $(v) $ holds, then $(vii)$ is implied by any of 
  \begin{itemize}
  \item[$(vii')$] $ \|\dc(y) - \dc(x)\|  \leq   C \big( \|\dc^{-1}(x)
     (y-x) \|^{\kappa_{\dc}} + \|\dc^{-1}(x)
     (y-x) \|^{m} \big), $  for all $x \in \mathbb{R}^d$ and some $C < \infty,$
\item[$(vii'')$]  The infinitesimal variance is H\"older continous and
  bounded above, i.e. $ \|\dc(y) - \dc(x)\|  \leq  C \|
     y-x \|^{\kappa_{\dc}}$ and  $\|\dc(x)\| \leq C$
    for all $x \in \mathbb{R}^d$ and some $C < \infty.$ 
  \end{itemize}
\end{remark}

\begin{remark} \label{remark_norms}
  Assumption \ref{assu_for_CIS_moments} is stated using the Euclidean
  matrix norm, however due to
  equivalency of norms in normed finite-dimensional vector spaces, it can be equivalently rephrased
  in terms of other matrix norms. In particular, in case of the $L^2$
  operator norm $\|\cdot\|_2$,
  that we shall use in the sequel, the latter condition of $(iv)$ translates into upper
  bounds on the largest eigenvalue of $\dmdder(x) $ and $(v)$ into the smallest
  eigenvalue of $\dm(x)$ being bounded away from~$0.$
\end{remark}

\begin{remark} \label{remark_sigmas}
  In Assumption \ref{assu_for_CIS_moments} the diffusion coefficient
  $\dc(x)$ can be replaced by any other version $\dc_*(x)$,
  s.t. $\dc_*(x)\dc_*^T(x) = \dm(x).$
\end{remark}

\subsection{Stability of weights}

In view of generic stability results of Section
\ref{sec:validity-cis}, it is enough to verify Assumptions
\ref{assu:momentunif_bdd} and  \ref{assu:momentunif_bdd_gen_f} and
conclude from Theorem \ref{th:tot_weight}. We shall follow this path and prove the following result.

\begin{theorem}\label{thm_weights}
  Consider Copycat CIS as detailed in Algorithm {\algcc}. The CIS weights
  $W_T$ are in $L^{p}$  for every time horizon $T$ and for every $0 < p < p^*$ 
under the following conditions from Assumption~\ref{assu_for_CIS_moments},
  \begin{itemize}
  \item[(a)]  Conditions $(i),$ $(ii),(iii),(iv), (v)$ and $(vi)$ 
	 with  $(\alpha - \kappa_{\dm}/2) p^* < \alpha,$ or
	 \item[(b)]  Conditions $(i),$ $(ii),(iii),(iv), (v)$ and $(vii)$
	 with  $(\alpha - \kappa_{\dc}/2) p^* < \alpha,$
  \end{itemize}
\end{theorem}

\begin{remark}
  Remarks~\ref{remark_ver_assu},~\ref{remark_norms}~and~\ref{remark_sigmas} apply.
\end{remark}

\begin{remark}
  For given $p$ and $\kappa_{\dm}$ or $\kappa_{\dc}$, one
  can choose $\alpha$ small enough s.t. $W_T \in L^p.$ In particular,
  if $\alpha \leq \kappa_{\dm}$ or $\alpha \leq \kappa_{\dc}$  respectively,
  then all moments exist.
\end{remark}

\begin{proof}[Proof of Theorem \ref{thm_weights}] Follows 
  from Theorem \ref{th:tot_weight} by verifying Assumption
  \ref{assu:momentunif_bdd}  via Proposition~\ref{prop_CC_incr_weight_unif_bdd} below and Assumption \ref{assu:momentunif_bdd_gen_f}    via Lemma \ref{lemma:tails}
    and Corollary \ref{cor:mom_gen_f} below.
\end{proof}

The next lemma deals with the tails of $n_T.$ 

\begin{lemma} \label{lemma:tails}
Let $\lambda(u)$ and $n_t$ be as in
\eqref{eq:lambda} and  \eqref{def_n_t} respectively.  Then
for every $c< 1$ and $t,$ there exists $C < \infty$ such
that
\begin{eqnarray}
  \label{eq:6}
  \mathbb{P}(n_t \geq j) & \leq &  C j^{-c\alpha j} \qquad \quad
  \textrm{for every} \quad j \in \mathbb{N}.
\end{eqnarray}
\end{lemma}
\begin{proof}

Without loss of generality assume $\delta=1$ in the definition of
$\lambda(u)$ and thus the density and cdf of $\Delta \tau_1$ writes \begin{eqnarray} \label{density_of_interarrival}
g_{\Delta \tau_1}(u) & = & \lambda(u)
\exp\{- \int_0^u \lambda(v)dv \} \;\;=\;\; {1 \over u^{1-\alpha}}
\exp\{ - {1 \over \alpha} u^{\alpha}\} \qquad  \textrm{and} \\
G_{\Delta \tau_1}(u)
& = & 1-
\exp\{- \int_0^u \lambda(v)dv \} \;\;=\;\; 1-
\exp\{ - {1 \over \alpha} u^{\alpha}\} \;\; \leq \;\;
\alpha^{-1} u^{\alpha},  \label{eqn:cdf_of_J}
 \end{eqnarray} respectively. 
For $a < \alpha^{1/\alpha},$ define $b: =
\alpha^{-1} a^{\alpha}. $ By \eqref{eqn:cdf_of_J}, the cdf
of $\Delta \tau_1$ is smaller than the cdf of $aB(1-b),$ where $B(1-b)$ is a
Bernoulli random variable with success probability $1-b$. Let
$NB(r,b)$ denote the negative binomial distribution of the
number of success in a sequence of Bernoulli trials
(with success probability $b$) before $r$ failures occur. We allow for
 $0<r \in
\mR ,$ resulting in the generalized negative binomial
distribution. Now note that for
every choice of $a < \alpha^{1/\alpha},$ the random variable $n_t$ is
stochastically smaller than $ \lfloor t/a \rfloor + 1 + NB( \lfloor
t/a \rfloor + 1 , b), $ and consequently, by stochastic monotonicity
in the first parameter between generalized negative binomials, $n_t$ is also stochastically smaller than $N_J(a):= t/a + 1+ NB(t/a + 1, b).$
The moment generating function of $N_J(a),$ say $M_{N_J(a)}(s),$
exists for $ s \leq - \log b$ and writes 
\begin{eqnarray}
  \label{eq:7}
  M_{N_J(a)}(s) & = & e^{s{t + a \over a}} ({1-b \over 1- b e^s})^{{t
      + a\over a}}.
\end{eqnarray}
Thus by stochastic ordering and then Markov inequality we obtain for every $s$
\begin{eqnarray}
  \label{ineq:Markov}
 \mathbb{P}(n(t) \geq j) \;\; \leq \;\;  \mathbb{P}( N_J(a) \geq j ) & \leq & e^{-sj} e^{s{t+a\over a}} ({1-b \over 1-
    b e^s})^{{t+a \over a}}.
\end{eqnarray}
We set $z:=e^s,$ take the $\log$ of  \eqref{ineq:Markov} and
differentiate it with respect to $z,$ to obtain that given $j > 1 + t/a ,$ the RHS of \eqref{ineq:Markov} is
minimized at $z = {a(j-1) - t \over ab j}. $ Using this and recalling
that $b = \alpha^{-1} a ^{\alpha},$ by elementary algebra
we arrive at
\begin{eqnarray}
  \nonumber 
   \mathbb{P}( N_J(a) \geq j ) & \leq & \alpha^{-j} \left( {aj - a -t
       \over a^{1+\alpha} j } \right)^{-j} 
\left( {(aj - a -t) \alpha
       \over a^{1+\alpha} j } \right)^{{t+a \over a}}
\left( {(\alpha - a^{\alpha} ) aj
       \over \alpha (a+t) } \right)^{{t+a \over a}}
\\ \label{eq:j_decomposition}
& = & \left( {\alpha \over a^{\alpha}} \right)^{-j} 
 \left( 1 - {1 \over
       j} - {t \over a j} \right)^{-j}
\left( {aj \over a+t}
     - 1 \right)^{{t+a \over a }} \left( {\alpha \over a^{\alpha}} - 1
   \right)^{{t+a \over a}}. 
\end{eqnarray}
The next step is to minimize the above with respect to $a$ provided it
is in the admissible range, i.e. $j > 1 +t/a,$ for
$j$ large enough. The
solution is not in closed form, however by
standard calculus one can see that for large $j$ the optimal value of
$a,$ say $a^*$ satisfies e.g. ${ \log \log j \over j } < a^* <  { \log j
  \over j }.$ In particular, since the inequality in
\eqref{eq:j_decomposition} holds for every $a$ from the admissible
range, we can set $a  = { \log j
  \over j }< \alpha^{1/\alpha}.$ Rearranging terms yields for $j$ large enough
\begin{eqnarray*}
  \label{eq:10}
 \mP ( N_J(a) \geq j ) & \leq & \left( {\alpha^{{1\over \alpha}} j \over \log j }  \right)^{-\alpha j}
\left( 1- {1 \over j} - {t \over \log j} \right)^{-j}
 \times \\
&& \qquad \qquad \qquad \times \left( {j \log j \over \log j +
     jt } -1 \right)^{{jt + \log j
     \over  \log j}} 
 \left( {\alpha j^{\alpha } \over (\log j)^{\alpha} } -1 \right)^{{jt + \log j
     \over  \log j}} \\ & \leq &
\left( {\alpha^{{1\over \alpha}} j \over \log j }  \right)^{-\alpha j}
\left( 1- {1 \over j} - {t \over \log j} \right)^{-j}
 \left( { \log j \over 2
     t } \right)^{{jt + \log j
     \over  \log j}} 
 \left( {\alpha^{1 \over \alpha} j \over \log j }\right)^{\alpha{jt + \log j
     \over  \log j}} 
\\ & \leq & j^{- (1-{t \over \log j} - {1 \over j})\alpha j} \times
(\log j)^{\alpha j + {jt + \log j
     \over  \log j}(1-\alpha) } \times \\ 
&&  \qquad \qquad \qquad
                                             \qquad  \times \alpha^{-j}  \left( 1- {1 \over
     j} - {t \over \log j} \right)^{-j} \left({\alpha
     \over 2t }\right)^{jt + \log j
     \over  \log j}
\\
  & =: & A(j) \times B(j) \times C(j). \\ & \leq & j^{-c\alpha
    j},\qquad \textrm{for
    every} \;\; c< 1  \;\; \textrm{and $\;  j \; $ large enough,}
\end{eqnarray*}
where the last inequality results from verifying assymptotics with $j \to \infty,$ namely
\begin{eqnarray*}
  \label{eq:8}
  {C(j) \over M^j} & \to & 0,  \qquad \textrm{for
    large enough} \quad M,  \\
  { B(j) \over (\log j)^{2j} } & \to & 0, \qquad \textrm{and} \\
  { A(j) \over j^{-c_0\alpha j}} & \to & 0, \qquad \textrm{for
    every} \quad c_0< 1.  
\end{eqnarray*}
Thus the proof is complete for arbitrary $c < 1$ by setting $c_0 = 1 - (1-c)/2.$
\end{proof}

The following weaker result is an immediate corollary from Lemma
\ref{lemma:tails} and verifies Assumption \ref{assu:momentunif_bdd_gen_f}.
\begin{cor}\label{cor:mom_gen_f}
 The random variable $n_T$ has everywhere
 finite moment generating function, i.e. $\ecis [\exp\{c\te{T}\}] <
 \infty$ for every $c \in \mathbb{R}.$
\end{cor}

Now we turn to establishing Assumption \ref{assu:momentunif_bdd} via verifying equation
\eqref{eq:momentunif_bdd_smpl}.

\begin{prop}\label{prop_CC_incr_weight_unif_bdd} Consider the Copycat version of the CIS algorithm for SDEs.
Under the following conditions from Assumption~\ref{assu_for_CIS_moments},
  \begin{itemize}
  \item[(a)]  Conditions $(i),$ $(ii),(iii),(iv),(v)$ and $(vi)$ 
	 with  $p(\alpha - \kappa_{\dm}/2) < \alpha,$ or
	 \item[(b)]  Conditions $(i),$ $(ii),(iii),(iv,(v))$ and $(vii)$
	 with  $p(\alpha - \kappa_{\dc}/2) < \alpha.$
  \end{itemize}
equation \eqref{eq:momentunif_bdd_smpl} holds with $p^*$ replaced by
$p$, i.e.  the $p-$th moment of the incremental
  weight is uniformly bounded.
\end{prop}

\begin{proof}
We start with some preliminary calculations useful for both cases.
Recall that for Copycat CIS for diffusions,
conditionally on $X_0 =
x,$ $X_{\tau_1} = y$ and $\tau_1 = u,$ for the incremental weight $H_1
= H_1(x,y,u),$ we have 

\begin{eqnarray} \nonumber
 \qquad \quad   H_1 & = & 1 + { (\tro - \pro)\prdcc{0}{x}{y}{u} \over
   \lambda(u) \prdcc{0}{x}{y}{u} } \\ \nonumber 
& = & 1 + {1\over \lambda(u)} \Bigg( {1\over 2} [\dm(y) - \dm(x)]:
  \qdder(x,y,u) \\ \nonumber && 
 + \Big[ \dmder(y)\one_{d\times 1} - \bb(y) + \bb(x)
  \Big] \cdot \qder(x,y,u) + {1 \over 2} \dmdder(y):\one_{d\times d} -
  \bder(y) \cdot \one_{d \times 1} \Bigg) \\ \label{eq:incremental_for_CC}
& = & 1+  {1\over \lambda(u)} \bigg( {1\over 2} D_2   + D_1  + D_0 \bigg), \qquad \textrm{where}  \\ \nonumber 
D_2 & = &  [\dm(y) - \dm(x)]:
  \qdder(x,y,u),  \\ \nonumber 
D_1 & = &  \big[ \dmder(y)\one_{d\times 1} - \bb(y) + \bb(x)
  \big] \cdot \qder(x,y,u) , \\ \nonumber 
D_0 & = &  {1 \over 2} \dmdder(y):\one_{d\times d} -
  \bder(y) \cdot \one_{d \times 1} .
\end{eqnarray}
Recall that in the above $\; : \;$ and $\; \cdot \;
$ stand for the Frobenius and Euclidean
inner products respectively. To rewrite $\qder(x,y,u)$ and
$\qdder(x,y,u),$ in a suitable form, we now use $z,$ the realisation of  the $d-$dimensional standard normal vector used to generate $y$ in the Copycat scheme. 
\begin{eqnarray}
 \nonumber
\qder(x,y,u) & = & - {\dm^{-1}(x) \over u}[y-x-ub(x)] \;\; = \;\; -
{\dc^{-T}(x) \dc^{-1}(x) \over u} \sqrt{u} \dc(x) z \;\; = \;\; -
{\dc^{-T}(x) \over \sqrt{u}} z , \\ \nonumber
\qdder(x,y,u) & = & \qder(x,y,u) \qder^T(x,y,u) - {\dm^{-1}(x) \over u} \;\; = \;\; {1 \over u} \dc^{-T}(x)(z z^T -I)\dc^{-1}(x)  .
\end{eqnarray}
For $A =
A_1 + \cdots + A_k,$ and arbitrary $c>0 ,$ we make repeated use of 
  $$|A|^c  =   |A_1 + \cdots + A_k|^c  \leq (k \max_i|A_i|)^c
 \leq  k^c (|A_1|^c + \cdots + |A_k|^c).$$ In particular, since by \eqref{eq:incremental_for_CC}
\begin{eqnarray}
  \label{eq:H_deco}
  |H_1|^p & \leq & 4^p \Big(1 + 2^{-p}{|D_2|^p \over \lambda^p(u)} +  {|D_1|^p
    \over \lambda^p(u)} + {|D_0|^p \over \lambda^p(u)} \Big),
\end{eqnarray}
we can deal with expectations $ \mE {|D_i|^p \over \lambda^p(U)},$
$i=2,1,0,$ separately. \\
Now assume $(a).$\\
Compute $D_2$ using first the cyclic property of trace, then
Cauchy-Schwartz for the Frobenius norm.
\begin{eqnarray} \nonumber 
  |D_2| & = & {1 \over u} \Big|\trace \Big( [\dm(y) - \dm(x)] \dc^{-T}(x)(z z^T -I)\dc^{-1}(x) \Big)\Big| \\ \nonumber
& = &  {1 \over u}\Big| \trace \Big( \dc^{-1}(x) [\dm(y) - \dm(x)] \dc^{-T}(x)(z z^T -I)\Big) \Big| \\ \nonumber 
& \leq &  {1 \over u} \| \dc^{-1}(x) [\dm(y) - \dm(x)]\| \|
\dc^{-T}(x)(z z^T -I) \| \;\; =: \;\; \diamondsuit_1.  
\end{eqnarray}
Next, use submultiplicativity of the Frobenius norm, uniform
boundedness of the $L^2$ operator norm of $\dc^{-1}(x)$ resulting from
Assumption~\ref{assu_for_CIS_moments}~$(v),$ and subsequently
Assumption~\ref{assu_for_CIS_moments}~$(vi).$ 
\begin{eqnarray} \nonumber 
\diamondsuit_1 & \leq & {1 \over u} \| \dc^{-1}(x) [\dm(y) - \dm(x)]\| \|
\dc^{-T}(x)\| \|z z^T -I \| \\ \nonumber 
& \leq & {C \over u} \| \dc^{-1}(x) [\dm(y) - \dm(x)]\| \|z z^T -I \| \\ \nonumber 
& \leq &  {C \over u} \Big( \| \dc^{-1}(x) (y-x)\|^{\kappa_{\dm}} + \|\dc^{-1}(x) (y-x)\|^m \Big)\|z z^T -I \| \\ \nonumber
& \leq & {C \over u} \Bigg(\| \dc^{-1}(x) (y-x -
ub(x))\|^{\kappa_{\dm}} +
u^{\kappa_{\dm}}\|\dc^{-1}(x)b(x)\|^{\kappa_{\dm}} \\ \nonumber 
&& \qquad \quad + \| \dc^{-1}(x) (y-x -
ub(x))\|^{m} +
u^{m}\|\dc^{-1}(x)b(x)\|^{m}  \Bigg) \|z z^T -I
\| \;\; =: \;\; \diamondsuit_2. \end{eqnarray}
Finally, use Assumption~\ref{assu_for_CIS_moments}~$(ii),$ and
realisation $z $  of the
generic standard normal vector used to generate $y$ to rewrite the above as
\begin{eqnarray} \nonumber
\diamondsuit_2  & \leq & {C \over u} \Big(\| \dc^{-1}(x)u^{1/2}
\dc(x)z\|^{\kappa_{\dm}} + u^{\kappa_{\dm}} + \| \dc^{-1}(x)u^{1/2} \dc(x)z\|^{m} + u^{m}\Big) \|z z^T -I \| \\
& = & {C \over u} \Big( u^{\kappa_{\dm}/2}\| z\|^{\kappa_{\dm}} +
u^{\kappa_{\dm}} +  u^{m/2}\| z\|^{m} + u^{m}\Big) \|z z^T -I \|. \label{D_2__under_d_conditional}
\end{eqnarray}
If $Z$ is a
$d-$dimensional standard normal, then for every $p$ there exists $C<
\infty,$ s.t. $\mE\|ZZ^T - I\|^p \leq C$ and
$\mE\big(\|Z\|^{p\kappa_{\dm}}\|ZZ^T - I\|^p\big) \leq C.$ Similarly
with $\kappa_{\dm}$ replaced by $m.$ Denote the
$d-$dimensional standard normal density by $\varphi(\cdot).$ Now recall the density of the interarrival time
\eqref{density_of_interarrival} and using
\eqref{D_2__under_d_conditional}, compute
\begin{eqnarray}
   \mE {|D_2|^p \over \lambda^p(U)} & \leq & C \int_{0}^T {u^{(p-1)(1-\alpha)} \over u^p
   } \exp\big\{-{1\over \alpha} u^{\alpha} \big\} \times  \nonumber \\
   &&  \qquad \times \bigg( \int_{\mR^d} \Big(
    u^{\kappa_{\dm}/2}\| z\|^{\kappa_{\dm}} +
u^{\kappa_{\dm}} +  u^{m/2}\| z\|^{m} + u^{m}\Big)^p
   \|z z^T -I \|^p \varphi(z)\d z\bigg) \d u \nonumber \\
 & \leq & C \int_{0}^T u^{\alpha(1-p) - 1} \big( u^{p\kappa_{\dm}/2} +
 u^{p\kappa_{\dm}} + u^{pm/2} +
 u^{pm} \big) \d u\qquad \nonumber \\ \label{eq:D_2_to_p}
 & \leq &  C \int_{0}^T u^{\alpha(1-p)
   - 1 + p\kappa_{\dm}/2} \d u \;\; \leq \;\; \mathcal{C}_2(p) \;\; < \;\; \infty,
\end{eqnarray}
since $m > \kappa_{\dm}$ and we assumed $p(\alpha - \kappa_{\dm}/2) < \alpha.$ 

To deal with the scalar product in $D_1$ use first that
$\sigma^{-1}(x)$ is the adjoint of $\sigma^{-T}(x),$ then
Cauchy-Schwarz for the Frobenius norm, and finally Assumptions~\ref{assu_for_CIS_moments}~$(iv)$ and
$(v),$ i.e. boundedness of $\|\dmder(x)\|$ and $\|\sigma^{-1}(x)\|.$
\begin{eqnarray}
 \nonumber
  |D_1| & \leq & \Big| \big[ \dmder(y)\one_{d\times 1} - \bb(y) + \bb(x)
  \big] \cdot {\dc^{-T}(x) \over \sqrt{u}} z \Big| \\ \nonumber
& \leq & u^{-1/2}\Big(\big|\dc^{-1}(x) \dmder(y) \one_{d\times 1}
\cdot z  \big| + \big| \dc^{-1}(x)(\bb(y)-\bb(x)) \cdot z \big|  \Big) \\ \nonumber
& \leq & C u^{-1/2}\Big(\|z\| + \| \dc^{-1}(x)(\bb(y)-\bb(x)) \| \|z\|
\Big) \;\; =: \;\; \clubsuit_1 \end{eqnarray} 
Next, use Assumption~\ref{assu_for_CIS_moments}~$(i)$ and $(ii)$ and
again $z$, the realisation of
the standard normal vector used to generate $y.$ 
\begin{eqnarray}\ \nonumber
\clubsuit_1 & \leq & C u^{-1/2}\Big(\|z\| + \big( 1+ \| \dc^{-1}(x)(y-x) \|^{\kappa_{\bb}} \big) \|z\|
\Big)  \\ \nonumber
& \leq & C u^{-1/2}\Big(\|z\| + \big( \| \dc^{-1}(x)(y-x - ub(x)) \|^{\kappa_{\bb}} +
u^{\kappa_{\bb}}\|\dc^{-1}(x)\bb(x)\|^{\kappa_{\bb}} \big) \|z\|
\Big) \\ \nonumber 
& \leq & C u^{-1/2}\Big(\|z\| + \big( u^{\kappa_{\bb}/2}\|z \| +u^{\kappa_{\bb}}
\big) \|z\| 
\Big) \\
& \leq &  C u^{-1/2} \Big(  \|z\| + u^{\kappa_{\bb}}
 \|z\| + u^{\kappa_{\bb}/2} \|z \|^2  \Big). 
\end{eqnarray}
Similarly as in \eqref{eq:D_2_to_p}, recalling \eqref{density_of_interarrival} and
noting that $\mE \|Z\|^p < \infty$ and $\mE \|Z\|^{2p} < \infty$ for
$Z \sim N(0, I),$ yields
\begin{eqnarray}
   \mE {|D_1|^p \over \lambda^p(U)} & \leq & C \int_{0}^T {u^{(p-1)(1-\alpha)} \over u^{p/2}
   } \exp\big\{-{1\over \alpha} u^{\alpha} \big\} \times \nonumber \\
   && \qquad \qquad \times \Bigg( \int_{\mR^d} \Big(
    \|z\| + u^{\kappa_{\bb}}
 \|z\| + u^{\kappa_{\bb}/2} \|z \|^2 \Big)^p
   \varphi(z)\d z\Bigg) \d u \nonumber \\
 & \leq & C \int_{0}^T u^{p(1/2 - \alpha) - 1 + \alpha } \big( 1+ u^{p\kappa_{\bb}/2} +
 u^{p\kappa_{\bb}}\big) \d u \qquad \nonumber \\ \label{eq:D_1_to_p}
 & \leq & C \int_{0}^T  u^{p(1/2 -
   \alpha) - 1 + \alpha } \d u \;\; \leq \;\; C_1(p) \;\; < \;\; \infty,
\end{eqnarray}
since $p(\alpha - \kappa_{\dm}/2) < \alpha$ implies $p(\alpha - 1/2) < \alpha.$

Finally, by  Assumptions~\ref{assu_for_CIS_moments}~$(iii)$ and~$(iv)$
and Cauchy-Schwarz, for some $C < \infty,$
\begin{eqnarray} \nonumber
  |D_0| &\leq &  {1 \over 2} \big| \trace \big(\dmdder(y) \one_{d\times d}\big) \big|+
  \big| \bder(y) \cdot \one_{d \times 1} \big|  \\ 
 & \leq &  \|\dmdder(y)\| \|\one_{d\times d}\| +\|\bder(y)\| \| \one_{d
   \times 1}\| \;\; \leq \;\; C_0 \;\; < \;\; \infty.
\end{eqnarray}
Hence 
\begin{eqnarray}
\mE{|D_0|^p \over \lambda^p(U)} & = & C \int_{0}^T u^{(p-1)(1-\alpha)}
\exp\big\{-{1\over \alpha} u^{\alpha} \big\} \Big( \int_{\mR^d} C_0^p 
   \varphi(z)\d z\Big) \d u \nonumber \\
& \leq & C \int_{0}^T u^{(p-1)(1-\alpha)} \d u \;\; \leq \;\; C \;\;
< \;\; \infty.
 \label{eq:D_0_to_p}  
\end{eqnarray}

Combining \eqref{eq:D_2_to_p}, \eqref{eq:D_1_to_p} and
\eqref{eq:D_0_to_p} yields the claim under assumption $(a)$.

Now assume $(b).$

The calculations involving $|D_1|^p$ and $|D_0|^p$ do not change. To
deal with $|D_2|$ note that
\begin{eqnarray}
\qquad  \dm(y) - \dm(x) & = & \big(\dc(y) - \dc(x) +
  \dc(x)\big)\big(\dc^T(y)-\dc^T(x) +\dc^T(x) \big) - \dc(x)\dc^T(x)
  \nonumber \\ & = & \big(\dc(y) - \dc(x)
  \big)\big(\dc^T(y)-\dc^T(x) \big)
   \label{eq:gamma_difference_into_sigma}
  \\ && \nonumber \qquad  \qquad 
 +  \dc(x)\big(\dc^T(y)-\dc^T(x)\big)  + 
\big(\dc(y) - \dc(x)\big)\dc^T(x). 
\end{eqnarray}
Now use \eqref{eq:gamma_difference_into_sigma} to deal with  $|D_2|$
in a similar manner that led to \eqref{D_2__under_d_conditional}. We
thus start with utilising the decomposition
\eqref{eq:gamma_difference_into_sigma}, next use cyclic property of
the trace, and submultiplicativity of the Frobenius norm to arrive at
\begin{eqnarray} \nonumber 
  |D_2| & = & {1 \over u} \Big|\trace \Big( [\dm(y) - \dm(x)] \dc^{-T}(x)(z z^T -I)\dc^{-1}(x) \Big)\Big| \\ \nonumber
& \leq  & {1 \over u} \Big|\trace \Big( \big(\dc(y) - \dc(x)
  \big)\big(\dc^T(y)-\dc^T(x) \big)
\dc^{-T}(x)(z z^T -I)\dc^{-1}(x) \Big)\Big| \\ \nonumber && \qquad \qquad \qquad \qquad 
+ 
{1 \over u} \Big|\trace
\Big(\dc(x)\big(\dc^T(y)-\dc^T(x)\big) \dc^{-T}(x)(z z^T
-I)\dc^{-1}(x) \Big)\Big|  \\ \nonumber && \qquad \qquad \qquad \qquad \qquad 
+
{1 \over u} \Big|\trace \Big(\big(\dc(y) - \dc(x)\big)\dc^T(x)\dc^{-T}(x)(z z^T
-I)\dc^{-1}(x) \Big)\Big|  \\ \nonumber 
& \leq  & {1 \over u} \| \dc^{-1}(x) \big(\dc(y) - \dc(x)
  \big) \| \| \big(\dc^T(y)-\dc^T(x) \big)
\dc^{-T}(x) \| \| z z^T -I \|  \\ \nonumber && \qquad \qquad \qquad \qquad 
+ 
{1 \over u} \| \big(\dc^T(y)-\dc^T(x)\big) \dc^{-T}(x)\| \|z z^T
-I \|  \\ \nonumber && \qquad \qquad \qquad \qquad \qquad 
+
{1 \over u} \| \dc^{-1}(x) \big(\dc(y) - \dc(x)\big)\| \|z z^T
-I\| 
 \\ \nonumber 
& =  & {1 \over u} \| \dc^{-1}(x) \big(\dc(y) - \dc(x)
  \big) \|^2 \| z z^T -I \| 
+
{2 \over u} \| \dc^{-1}(x) \big(\dc(y) - \dc(x)\big)\| \|z z^T
-I\| \\ \nonumber & :=  &
\spadesuit_1 \end{eqnarray} 
Next, use Assumption~\ref{assu_for_CIS_moments}~$(vii)$ and $(ii)$
subsequently. The last step will be to rewrite the bound using~$z,$
the realisation of standard normal used for generating $y.$
\begin{eqnarray}\ \nonumber 
\spadesuit_1 & \leq  & {C \over u} \Big( \| \dc^{-1}(x)
(y-x)\|^{2\kappa_{\dc}} + \| \dc^{-1}(x) (y-x)\|^{2m} \Big)
\| z z^T -I \| \\ \nonumber
&& \qquad + 
{C \over u}  \Big( \| \dc^{-1}(x) (y-x)\|^{\kappa_{\dc}} + \|
\dc^{-1}(x) (y-x)\|^{m} \Big) \|z z^T
-I \|  \\ \nonumber
& \leq & {C \over u} \Big(\| \dc^{-1}(x) (y-x -
ub(x))\|^{2\kappa_{\dc}} +
u^{2\kappa_{\dc}}\|\dc^{-1}(x)b(x)\|^{2\kappa_{\dc}} \\ \nonumber &&
\qquad \qquad 
+ \| \dc^{-1}(x) (y-x -
ub(x))\|^{2m} +
u^{2m}\|\dc^{-1}(x)b(x)\|^{2m}
\Big) \|z z^T -I
\| \\ \nonumber
&  & \qquad + {C \over u} \Big(\| \dc^{-1}(x) (y-x -
ub(x))\|^{\kappa_{\dc}} +
u^{\kappa_{\dc}}\|\dc^{-1}(x)b(x)\|^{\kappa_{\dc}} \\ \nonumber &&
\qquad \qquad \qquad 
+ \| \dc^{-1}(x) (y-x -
ub(x))\|^{m} +
u^{m}\|\dc^{-1}(x)b(x)\|^{m}
\Big) \|z z^T -I \| 
\\ \nonumber
& \leq & {C \over u} \Big(\| \dc^{-1}(x)u^{1/2}
\dc(x)z\|^{2\kappa_{\dc}} + u^{2\kappa_{\dc}} + \| \dc^{-1}(x)u^{1/2}
\dc(x)z\|^{\kappa_{\dc}} + u^{\kappa_{\dc}} \\ \nonumber && \qquad 
+ \| \dc^{-1}(x)u^{1/2}
\dc(x)z\|^{2m} + u^{2m} + \| \dc^{-1}(x)u^{1/2} \dc(x)z\|^{m} + u^{m}
\Big) \|z z^T -I \| \\
& = & {C \over u} \Big( u^{\kappa_{\dc}}\| z\|^{2\kappa_{\dc}} +
u^{\kappa_{\dc}/2}\| z\|^{\kappa_{\dc}} + u^{2\kappa_{\dc}}  +
u^{\kappa_{\dc}}  \nonumber \\ && \qquad 
+ u^{m}\| z\|^{2m} +
u^{m/2}\| z\|^{m} + u^{2m}  +
u^{m} 
\Big) \|z z^T -I \|.
 \label{D_2__under_e_conditional}
\end{eqnarray}
The result of \eqref{D_2__under_e_conditional} is then used in a
computation analogous to \eqref{eq:D_2_to_p}. Remembering that $m > \kappa_{\dc}$ it yields 
\begin{eqnarray}
  \label{eq:D_2_to_p_under_e}
   \mE {|D_2|^p \over \lambda^p(U)} & \leq & \dots \;\; \leq \;\;  C \int_{0}^T u^{\alpha(1-p)
   - 1 + p\kappa_{\dc}/2} \d u \;\; \leq \;\; \mathcal{C}_2(p) \;\; < \;\; \infty,
\end{eqnarray}
provided $p(\alpha - \kappa_{\dc}/2) < \alpha.$

Combining \eqref{eq:D_2_to_p_under_e}, \eqref{eq:D_1_to_p} and
\eqref{eq:D_0_to_p} yields the claim under assumption $(b)$. 

\end{proof}







\section{Extensions}
\label{sec:extensions-cis}

This section deals with boosting the algorithmic performance of CIS
and introduces two modifications. The first generalises the basic
algorithm by allowing the use of proposal distributions other than
$q$. An optimal distribution, in the sense of weight variability
minimisation, is identified when interest lies in
simulation/estimation at a fixed time $T$. This extension also allows
us to improve transition density estimation by using proposal
processes that are guided towards a specific value at the terminal
time. The second modification deals with estimation as $T$ increases
and enhances stability by using a batch implementation of CIS, which
monitors the performance of the algorithm through time and controls
the variability of the weights via resampling.

\subsection{Proposal densities and optimal implementation}

\label{sec:opt_dens}

Recall that the CIS algorithm described in Algorithm {\algcc}
updates the process at each event time using proposals from $q$, as
shown in Step \ref{item:cis_up_draw}. We now modify this step so that $\x{\nk}$ at each
event time is drawn from a generic density, denoted by $\pg{}{\cdot}$
where $\xi$ is an abstract parameter. To preserve the unbiasedness
of the algorithm, an importance sampling correction term is added to
the incremental weight so that Step \ref{item:cis_up_wei} becomes
\begin{align*} \w{\nk}=\w{\ck}
\mrwcc{\ck}{\x{\ck}}{\x{\nk}}{\dt{k}}\frac{\prdcc{\ck}{\x{\ck}}{\x{\nk}}{\dt{k}}}{\pg{\ck}{\x{\nk}}}.
\end{align*}
The above expression allows us to identify a proposal distribution
which minimises the variance of $\w{\nk}$. It is trivial to show that the
optimal $g_{\xi}$ is given by
\begin{align}\label{eq:opt_dens} \pg{\ck}{\x{\nk}}\propto
  |\mrwcc{\ck}{\x{\ck}}{\x{\nk}}{\dt{k}}|\prdcc{\ck}{\x{\ck}}{\x{\nk}}{\dt{k}}.
\end{align}
The optimal density will typically be intractable but it can guide the choice
of a suitable approximation, as illustrated later in Section
\ref{sec:sv_model}.

\subsection{Transition density estimation and guided CIS}
\label{sec:cond_cis}

When interest lies in estimating the transition density $p(x_0,
x_{\th}, \th)$, the accuracy of CIS can be improved by choosing $g$ such
that the proposed process is guided towards $x_{\th}$. A process with
this property which can be easily simulated is the so-called modified
Brownian bridge \citep{dur:gal:2002}. If $g_{bb}( x_s, x_t; x_T,
\alpha)$ denotes the transition density of a scaled Brownian bridge,
i.e., for $0\leq s\leq t\leq T$,
\begin{align*} 
g_{bb}( x_s, x_t; x_T, a) = \gd\lcb x_t;
\frac{ x_s(T-t)+ x_T(t-s)}{T-s}, a\frac{(T-t)(t-s)}{T-s} \rcb,
\end{align*}
then we update the proposal at event time of CIS using $g_{bb}\{ \x{k},
\x{k+1}; \xx{t},\dm(\xx{k})\}$. Notice that the density incorporates
information from the end point through the mean and drifts the process
towards $x_t$. We refer to a CIS algorithm using the above proposal
as guided CIS (GCIS).

By recalling that $\tau_k,\,i=0,1,\ldots,\te{t}$ denote the simulated
time points in a single run of GCIS up to time $t$, then the
transition density estimator is given by
\begin{align*}
\prdcc{\te{\th}}{\x{t}}{\xx{\th}}{\th-\iat{t}}\prod_{k=\istart}^{\iend}\frac{\prdcc{\ck}{\x{\ck}}{\x{\nk}}{\dt{k}}}{g_{bb}\{\x{\ck}, \x{\nk}; \xx{\th},\dm(\x{\ck})\}}\mrwcc{\ck}{\x{\ck}}{\x{\nk}}{\dt{k}}.
\end{align*}
An interesting link to the transition density estimator of Durham and Gallant
\cite{dur:gal:2002} (henceforth denoted by DG) arises here. In
particular, notice that after removing the incremental weight terms from the
above quantity the density estimator only differs from a standard DG
density estimator which imputes $\te{t}$ points in that the partition of
the time interval is neither equidistant nor deterministic. This
allows us to give an intuitive explanation to the incremental weights
$\rho_{\theta_{\ck}}$ as bias correction terms to the subtransition
densities $q_{\theta_{\ck}}$.

\subsection{Resampling}\label{sec:resampling}

The variance of estimates obtained by sequential importance samplers
typically increases with time, implying that the approximations can be
arbitrarily poor when $\th$ is large, see for example
\cite{kong:liu:1994}. An established approach which attempts to
improve the quality of the estimates over large time horizons is based
on resampling procedures, see \cite{douc:cappe:moulines:2005} for a
recent review. The idea is to monitor the quality of the approximation
through time using a cloud of $N$ particles and rejuvenate them, when
necessary, by eliminating particles with low importance weights. In
this section, we adapt these ideas to the CIS framework and develop
two algorithms which incorporate a resampling step. We show
numerically that both algorithms improve greatly upon the basic CIS
algorithm presented in the previous sections when $\th$ is large. The
first version has a ${\cal O}(N)$ computational complexity but its
performance deteriorates superlinearly with $\th$. The second is
inspired by the  proposed method of \cite{doucet:score}, has
an ${\cal O}(N^2)$ cost and its performance degrades linearly with
$\th$.

Let $M\in{\mathbf Z^{+}}$ and denote by $t_j=j\th/M,\, j=0,1,\ldots,M$ a
partition of the time interval $[0,\th]$ into $M$ subintervals. The
algorithms that we propose use an initial cloud of $N$ particles and
propagate them through all times $t_j$ using CIS. At every time $t_j$
we have available the set of weights $\{\prt{w}{j}{i}\}_{i=1}^N$,
whose variability is assessed using an effective sample size (ESS), defined in step 1 of Algorithm 6 below. If the
ESS is small then the particles are resampled according to their
weights and only the surviving ones are propagated to the next time
$t_{j+1}$. We employ resampling based on $\{|\prt{w}{j}{i}|\}_{i=1}^N$
since the plain weights returned by CIS are not necessarily
positive. This may weaken the benefits of the resampling procedure,
since it can occasionally lead to resampling in areas where the
variability of the weights is large, rather than where the mean weight
is large. In the algorithms presented below, we denote the output of a
CIS algorithm with input $x$ and terminal time $t$ by the triplet
$(s,y,w)$ where $s$ is the last Poisson event before time $t$, $y$
the value of the process at that time and $w$ the output weight, and
write $(s,y,w)\sim \CIS(x,t)$. The first version of CIS with
resampling, termed CIS-R1, is presented in Algorithm {\algcisra}.

\begin{description}

\item[{\bf Algorithm {\algcisra}}:] {\it CIS algorithm with resampling} (CIS-R1).

\item[{\it Input}:] $\th$ a time to stop, $M$ a discretisation level, $C$
  an ESS threshold and $\xx{0}$ the initial value.

\item[{\it Initialise}:] Set $\prt{\algx}{0}{i}=\xx{0},\,\prt{s}{0}{i}=0,\,\prt{w}{0}{i}=1,\,1\leq i\leq N$ and $t_j=j\th/M,\,1\leq j\leq M-1$.
  
\item[\mdseries For $0\leq j\leq M-1,\,1\leq i\leq N$ perform the following steps] 

\item[{\it Step 1}:] Calculate the effective sample size of
  $\{|\prt{w}{j}{i}|\},\,\ESS=a_j^2/\sum_{i=1}^N (
  w_j^{(i)})^2$, where $a_j=\sum_{i=1}^N | w_j^{(i)}|$. If $\ESS<C$
  then sample $k_{i,j}\sim p(k)\propto|\prt{w}{j}{k}|,\,k=1,\ldots,N$
  and set $\prt{\phi}{j}{i} = \sign( \prt{w}{j}{k_{i,j}})
  a_j/N$; else set $k_{i,j}=j$ and $\prt{\phi}{j}{i}=\prt{w}{j}{i}$.

\item[{\it Step 2}:] Draw $(\prt{\tau}{j+1}{i},\,\prt{\algx}{j+1}{i},\,\prt{\bar{w}}{j+1}{i})\sim\CIS(\prt{\algx}{j}{k_{i,j}},\,t_{j+1}-\prt{s}{j}{k_{i,j}})$, set $\prt{s}{j+1}{i}=\prt{s}{j}{k_{i,j}}+\prt{\tau}{j+1}{i}$, and compute the weights $\prt{w}{j+1}{i} = \prt{\phi}{j}{i} \prt{\bar{w}}{j+1}{i}$.
  
\end{description}

Each step requires ${\cal O}(N)$ computations, therefore the
total complexity of CIS-R1 is ${\cal O}(N)$. When resampling occurs in
Step 1, the weights themselves are updated to $\prt{\phi}{j}{i} = a_j
\sign( \prt{w}{j}{k_{i,j}})/N$ to ensure the unbiasedness of the
algorithm. This also implies that surviving particles retain their weight
sign after resampling and CIS-R1 keeps propagating particles with
negative weights. 

A more efficient approach, albeit more computational expensive, is
presented below. The idea is based on simulating the process at
resampling times $t_j$ and assigning a numerically more stable weight
to each resampled particle. The algorithm,
termed CIS-R2, is presented below; for simplicity, we suppress the dependence of $q$ on $\theta$.

\begin{description}

\item[{\bf Algorithm \algcisrb}:] {\it Extended CIS algorithm with resampling} (CIS-R2).

\item[{\it Input}:] $\th$ a time to stop, $M$ a discretisation level, $C$
  an ESS threshold and $\xx{0}$ the initial value.

\item[{\it Initialise}:] Set $\prt{\algx}{0}{i}=\xx{0},\,\prt{s}{0}{i}=0,\,\prt{w}{0}{i}=1,\,1\leq i\leq N$ and $t_j=j\th/M,\,1\leq j\leq M-1$.
  
\item[\mdseries For $0\leq j\leq M-1,\,1\leq i\leq N$ perform the following steps] 

\item[{\it Step 1}:] Calculate the effective sample size of
  $\{|\prt{w}{j}{i}|\},\,\ESS=a_j^2/\sum_{i=1}^N (w_j^{(i)})^2$, where $a_j=\sum_{i=1}^N | w_j^{(i)}|$. If $\ESS<C$ then
  sample $\prt{\bx}{j}{i}\sim {\bar p_j}(\cdot)$, set
  $\prt{s}{j}{i}=t_j$ and $\prt{\phi}{j}{i}={\hat p}_j(\prt{\bx}{j}{i})/{\bar
    p}_j(\prt{\bx}{j}{i})$, where
        \begin{align}
          {\hat p_j}(y) &=
          \frac{1}{N}\sum_{k=1}^N\prt{w}{j}{i}\prdccn{}{\prt{\algx}{j}{i}}{y}{t_j-\prt{s}{j}{i}},\label{eq:ptilde}\\
          {\bar p_j}(y) &=
          a_j^{-1}\sum_{i=1}^N|\prt{w}{j}{i}|\prdccn{}{\prt{\algx}{j}{i}}{y}{t_j-\prt{s}{j}{i}}\label{eq:pbar},
        \end{align}

and set $\prt{\algx}{j}{i}=\prt{\bx}{j}{i}$; else set $\prt{\phi}{j}{i}=\prt{w}{j}{i}$.

\item[{\it Step 2}:] Draw $(\prt{\tau}{j+1}{i},\,\prt{\algx}{j+1}{i},\,\prt{\bar{w}}{j+1}{i})\sim\CIS(\prt{\algx}{j}{i},\,t_{j+1}-\prt{s}{j}{i})$, set $\prt{s}{j+1}{i}=\prt{s}{j}{i}+\prt{\tau}{j+1}{i}$, and compute the weights $\prt{w}{j+1}{i} = \prt{\phi}{j}{i} \prt{\bar{w}}{j+1}{i}$.
  
\end{description}

Resampling is still based on the absolute value of the weights and is
embedded in sampling from ${\bar p_j}$. Contrary to CIS-R1, when
resampling occurs at time $t_j$ the particles are also propagated to
that time. The weight given to each resampled particle
$\prt{\phi}{j}{i}$ is given by the ratio ${\hat p}_j/{\bar p}_j$ where
${\hat p}_j$ is the unbiased transition density estimator at time
$t_j$ and ${\bar p}_j$ is a mixture of our proposal
distributions. After resampling, the weighted particle set at time
$t_j$, $\{\prt{\bx}{j}{i},\prt{\phi}{j}{i}\}_{i=1}^N$ provides an
approximation to the density of the target process at time $t_j$ via
\begin{align*} {\hat \pi}_j(dy) =
\sum_{i=1}^N\prt{\phi}{j}{i}\delta_{\prt{\bx}{j}{i}}(dy)
\end{align*}
where $\delta_x(dy)$ denotes the Dirac delta mass located at $x$. By
conditioning on
$\{\prt{s}{j}{i},\,\prt{\algx}{j}{i},\,\prt{w}{j}{i}\}_{i=1}^N$, it is
trivial to show that ${\hat \pi}_j/N$ is an unbiased estimator of
${\hat p}_j$, and thus CIS-R2 retains the desired unbiased property.

Evaluating each of $\{\prt{\phi}{j}{i}\}_{i=1}^N$ requires ${\cal
  O}(N)$ calculations, hence the overall complexity of CIS-R2 is
${\cal O}(N^2)$. However, these weights will typically be numerically
more stable than that of CIS-R1 since they are calculated by
marginalising out the auxiliary variables up to time $t_j$; see also
\cite{doucet:score} for a similar implementation. Numerical evidence in
Section \ref{sec:simulation_resampling} suggest that CIS-R2 degrades only
linearly with $T$ whereas CIS-R1 superlinearly.


\section{Numerical illustrations}\label{sec:numericals}

In this section, we investigate numerically the performance of the CIS
algorithm under various implementation schemes and compare it to two
existing approaches. The first is the well known transition density
estimator of \cite{dur:gal:2002} which is based on discrete-time
approximations of the diffusion process and returns biased
estimates. The second is an unbiased approach introduced by Wagner
\cite{wagner:1989}, an overview of which is given below. All methods
under consideration are based on a finite imputation (either random or
deterministic) of points, thus allowing us to compare the approaches
in a fair manner on the basis of a fixed average number $K$ of
simulated values; that is, the number of trajectories $N$ multiplied
by the average number of points $M$ each trajectory is evaluated. We
refer to $K$ as the computational cost.

\subsection{Overview of Wagner's approach}
\label{sec:wagner}

The approach by Wagner (denoted by WGR hereafter) is a method which
estimates unbiasedly any quantities of interest by using unbiased
estimates of the transition density. The key of this method is the
representation of the unknown transition density via the following
integral equation
\begin{align*}
  \trd{\xx{0}}{\xx{t}}{t} = \wtrd{\xx{0}}{\xx{t}}{t} +
  \int_0^t\int_{\mR^d}\trd{\xx{0}}{\y}{s}{\cal C}(\y,\xx{t},t-s)\d \y
  \d s,
\end{align*}
where $\wtrd{\xx{0}}{\xx{t}}{t}$ is the transition density of a tractable
process and ${\cal C}(\y,\xx{t},t-s)$ is an analytically available
quantity which involves the Kolmogorov's backward operator \citep[for
more details see][]{wagner:1989}. By successive substitutions of the
transition density in the integral equation we obtain an infinite
expansion which is intractable but can be estimated unbiasedly using
importance sampling. In particular, if $t_0=t$ and $q_u(t,s),\,t>s$ is
a user-specified transition kernel, then the quantity
\begin{align*}
  \frac{\wtrd{\xx0}{\xx{t_n}}{t_n}}{p_u(t_n)}\prod_{k=1}^n\frac{{\cal
      C}(\xx{t_k},
    \xx{t_{k-1}},t_{k-1}-t_k)}{g_{bb}\{\xx0,\xx{t_k};\xx{t_{k-1}},\dm(\xx0)\}q_u(t_{k-1},t_k)[1-p_u(t_{k-1})]}
\end{align*}
is an unbiased estimator of the transition density. The algorithm
creates a random partition of the time interval by repeatedly
simulating time points $t_k\mid t_{k-1}\sim q_u(t_{k-1},t_k)$ and
drawing $\xx{t_k}\mid\xx0,\xx{t_{k-1}}$ according to Brownian bridge
dynamics. The function $p_u(t)$ is an absorption probability which
terminates the simulation and $n$ is the total number of simulated
time points.
Wagner proposes using
\begin{align*}
  p_u(t) =
  \lcb\sum_{m=0}^{\infty}\frac{\pin^m\Gamma(\pex)^mt^{m\pex}}{\Gamma(m\pex+1)}\rcb^{-1},\quad
q_u(t,s)=\frac{\pin\ind_{(0,t)}(s)(t-s)^{\pex-1}p_u(t)}{p_u(s)[1-p_u(t)]},
\end{align*}
where $\pin,\pex>0$ are specified by the user. 

Notice that when $\pex=1$ the above quantities can be calculated
analytically and the time points $\{t_k\}$ correspond to arrival times of
a simple homogeneous Poisson process of rate $\pin$ in $[0,t]$. We
refer to this particular implementation as WGR1. However, when
$\pex\neq 1$, calculating these quantities without sacrificing the
unbiasedness is not straightforward. 

For an objective comparison with CIS, we propose simulating the
time points $\{t_k\}$ using a process similar to that of CIS. In particular,
choosing
\begin{align*}
  p_u(t) = \exp\lcb-\frac{\gamma t^{\pex}}{\pex}\rcb, \quad q_u(t,s) = \frac{\gamma\ind_{(0,t)}(s)(t-s)^{\pex-1}p_u(t-s)}{1-p_u(t)},
\end{align*}
results in $\{t_k\}$ being a simple transformation of $\{\tau_k\}$ as
$t_k=t-\tau_k$, where $\{\tau_k\}$ are the time points simulated in a
CIS algorithm. We term such an implementation by WGR2. Empirical
results in the subsequent section suggest that WGR2 can
perform substantially better than WGR1.






\subsection{Stochastic volatility model}
\label{sec:sv_model}

We consider a simple bivariate diffusion process, solution to
\begin{align}
  \d X_{1,t} &= -\frac{\sigma_1^2}{2}\tanh (X_{1,t})\d t + \sigma_1\d B_{1,t},\notag \\
  \d X_{2,t} &= \sigma_2 [2 + \tanh(X_{1,t})]\d B_{2,t},\label{eq:sv_model}
\end{align}
where $B_1$ and $B_2$ are two independent standard
Brownian motions and $\sigma_1,\sigma_2>0$. We refer to the model as
SV. The first coordinate of the system is an ergodic process with
invariant mean $0$ while $\sigma_1$ governs the speed at which the
process returns to the invariant mean. It is straightforward to check
that the drift and diffusion matrix of this model satisfy conditions
for a valid CIS implementation. Despite its
fairly simple nature, \eqref{eq:sv_model} cannot be transformed to a diffusion with constant volatility, 
and therefor unbiased estimation cannot be performed using any variant of the Exact Algorithm.

\subsubsection{Implementation considerations}

The purpose of this section is to illustrate the effect that the
choice of the Poisson rate and proposal process can have on the
efficiency of the CIS algorithm. We employ CIS to estimate the mean
vector of \eqref{eq:sv_model} at terminal time $T=1$ using four
different implementations. The first is the copycat CIS as described
in Algorithm {\algcc} and estimates the integral in \eqref{eq:unbias}
by Monte Carlo. The second is implemented as the first but sets a
constant Poisson rate $\lambda(\tau)=\pin$; we refer to this estimator
as CIS$_{con}$. The third attempts to optimise the first by sampling
from the optimal density in \eqref{eq:opt_dens} and evaluating the
integral in \eqref{eq:unbias} analytically; we refer to this estimator
as CIS$_{opt}$. Simulating directly from the optimal density is not
straightforward; instead, we use an approximation where $x_{1,t}$ is
sampled from an Euler density as implied by \eqref{eq:sv_model}, while
the conditional distribution of $x_{2,t}$ given $x_{1,t}$ is
analytically tractable and is sampled using rejection
sampling. Finally, the fourth is implemented as the first but adapts
only the drift functional at every Poisson event; we refer to this
estimator as CIS$_{nc}$. To ensure a fair comparison, the value of
$\pin$ for each estimator is selected such that their computational
costs are comparable.

Figure \ref{fig:wrong_imp} shows the kernel density estimates of the
estimators, obtained by $1000$ independent realisations. Notice that
CIS$_{con}$ produces extreme estimates and its distribution is heavy
tailed. Both of these issues are remedied by using a nonconstant
Poisson rate, as illustrated by the distribution of the
CIS. CIS$_{opt}$ outperforms substantially the other methods while the
estimates obtained from CIS$_{nc}$ were highly unstable and exhibited
very heavy tails (not shown here).

\begin{figure}
\scriptsize
  \psfrag{CISp}{\hspace{1ex}CIS}
  \psfrag{CISnc}{\hspace{1ex}CIS$_{con}$}
  \psfrag{CISop}{\hspace{1ex}CIS$_{opt}$}
 \subfloat[]{\includegraphics[scale=0.42]{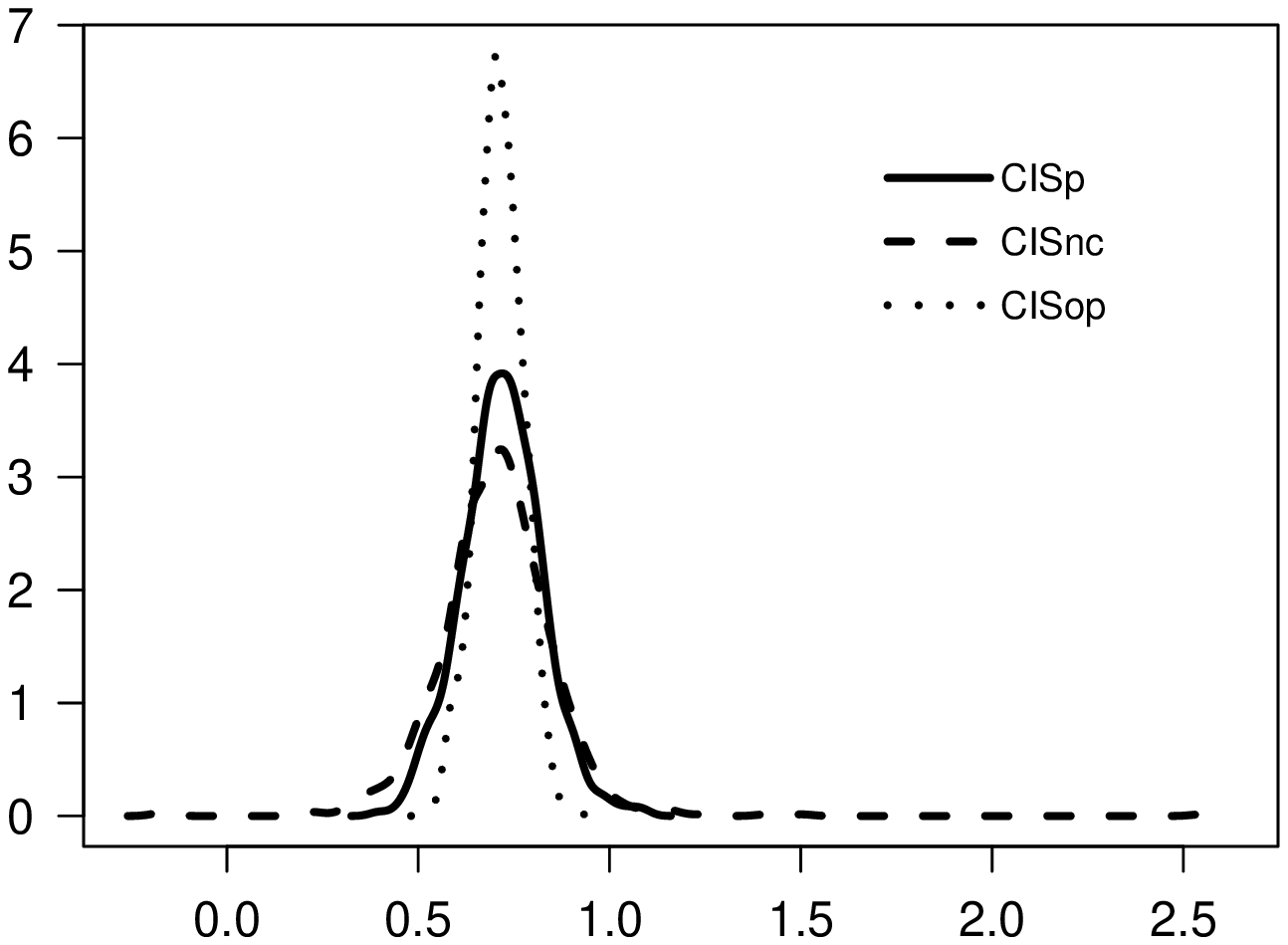} \label{fig:wrong_impa}}%
\subfloat[]{\includegraphics[scale=0.42]{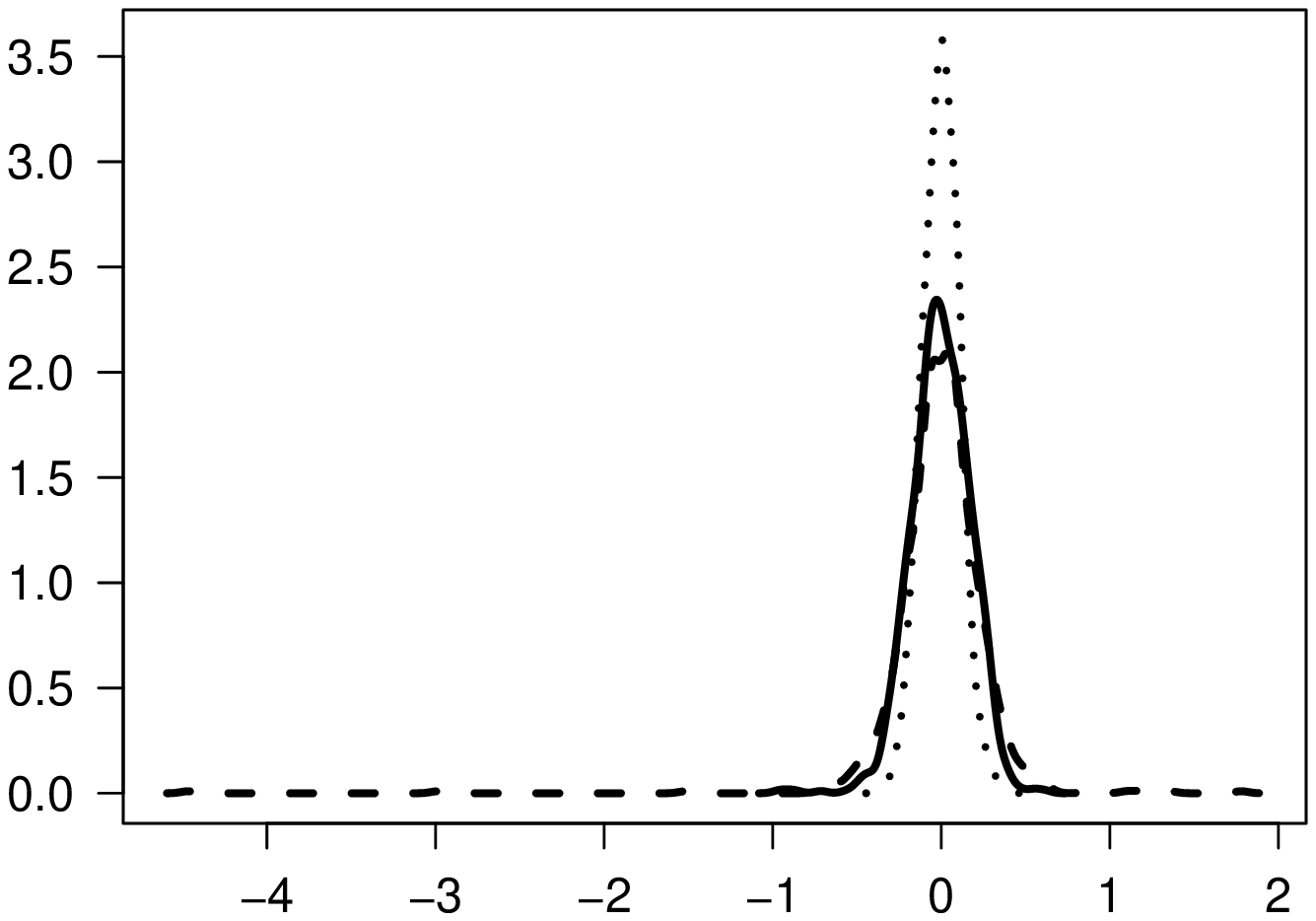}   \label{fig:wrong_impb}}
  \caption[]{The SV model with $\XX{0}=(1,0)'$ and $(\sigma_1,
    \sigma_2)=(1, 1/2)$. Kernel density estimates of estimators of
    \subref{fig:wrong_impa} $\mE(\XX{1,1})$ and \subref{fig:wrong_impb}
    $\mE(\XX{2,1})$ using CIS, CIS$_{opt}$ and CIS$_{con}$.}
  \label{fig:wrong_imp}
\end{figure}

\subsubsection{Estimation over large time intervals}
\label{sec:simulation_resampling}

We now assess the performance of CIS over large time intervals,
illustrate the effect of resampling using the two approaches developed
in Section \ref{sec:resampling} and compare CIS to the WGR
approach. We consider various values for the time horizon from $t=1/2$
to $t=10$ with stepsize $1/2$.

We begin by comparing the plain CIS, WGR1 and WGR2. The CIS algorithm
is employed with $N=1000$ whilst the Monte Carlo size of the other two
methods is tuned such that their computational cost are comparable to
that of CIS. Figure \ref{fig:inc_dt_1} illustrates the performance of
these three approaches. Notice that the approximation error of CIS is
substantially smaller than that of WGR1, and comparable to that of
WGR2. The WGR1 approach was highly unstable and generated many extreme
estimates. For this reason, and for this comparison only, we have
assessed the performance of the estimators using the median absolute
deviation (MAD) which is robust to outliers.

Figures \ref{fig:nores}-\ref{fig:res_nsq} present the root mean square
(RMSE) of the plain CIS along with its two resampling modifications,
CIS-R1 and CIS-R2. For an empirical illustration of the rate at which
the error increases, a linear, quadratic and cubic fit overlay the
plots. The plain CIS algorithm exhibits an approximation error which
increases cubically with time, and is outperformed significantly by
CIS-R1, whose error grows quadratically. Finally, at the expense of a
${\cal O}(N^2)$ computational cost, the performance of CIS-R2 seems to
deteriorate only linearly with time, implying that for sufficiently
large $T$, CIS-R2 will always outperform the other two methods.
%

\begin{figure}[ht]
  \centering
  \scriptsize
  \psfrag{cisnr}{CIS \footnotemark[1]}
  \psfrag{wgr1}{WGR1 \footnotemark[2]}
  \psfrag{wgr2}{WGR2 \footnotemark[3]}
  \psfrag{cisres}{CIS-R1 $^a$}
  \psfrag{cispres}{CIS-R2 $^a$}
  \subfloat[]{\includegraphics[scale=0.42]{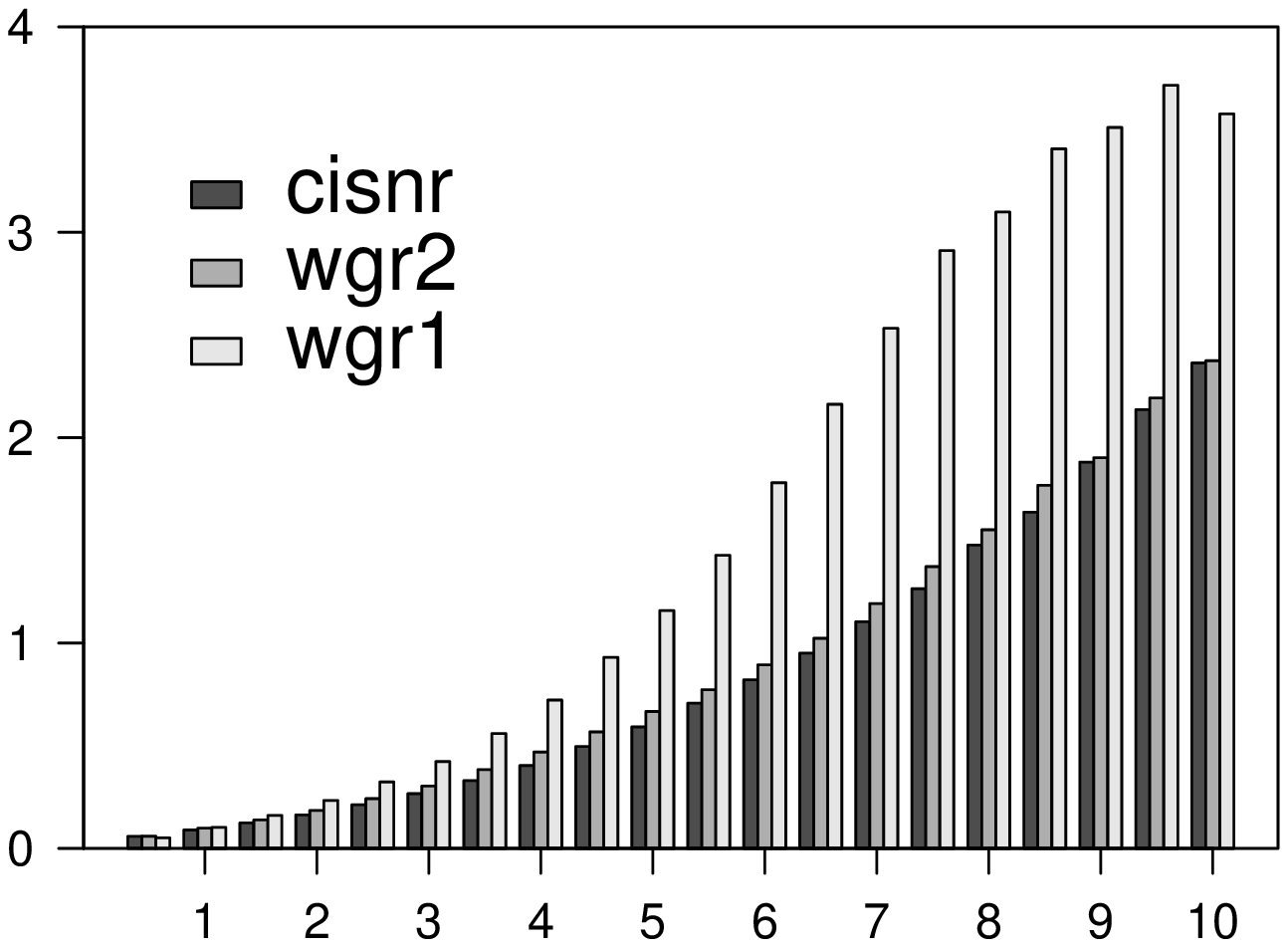}
  \label{fig:inc_dt_1}}
\subfloat[]{\includegraphics[scale=0.42]{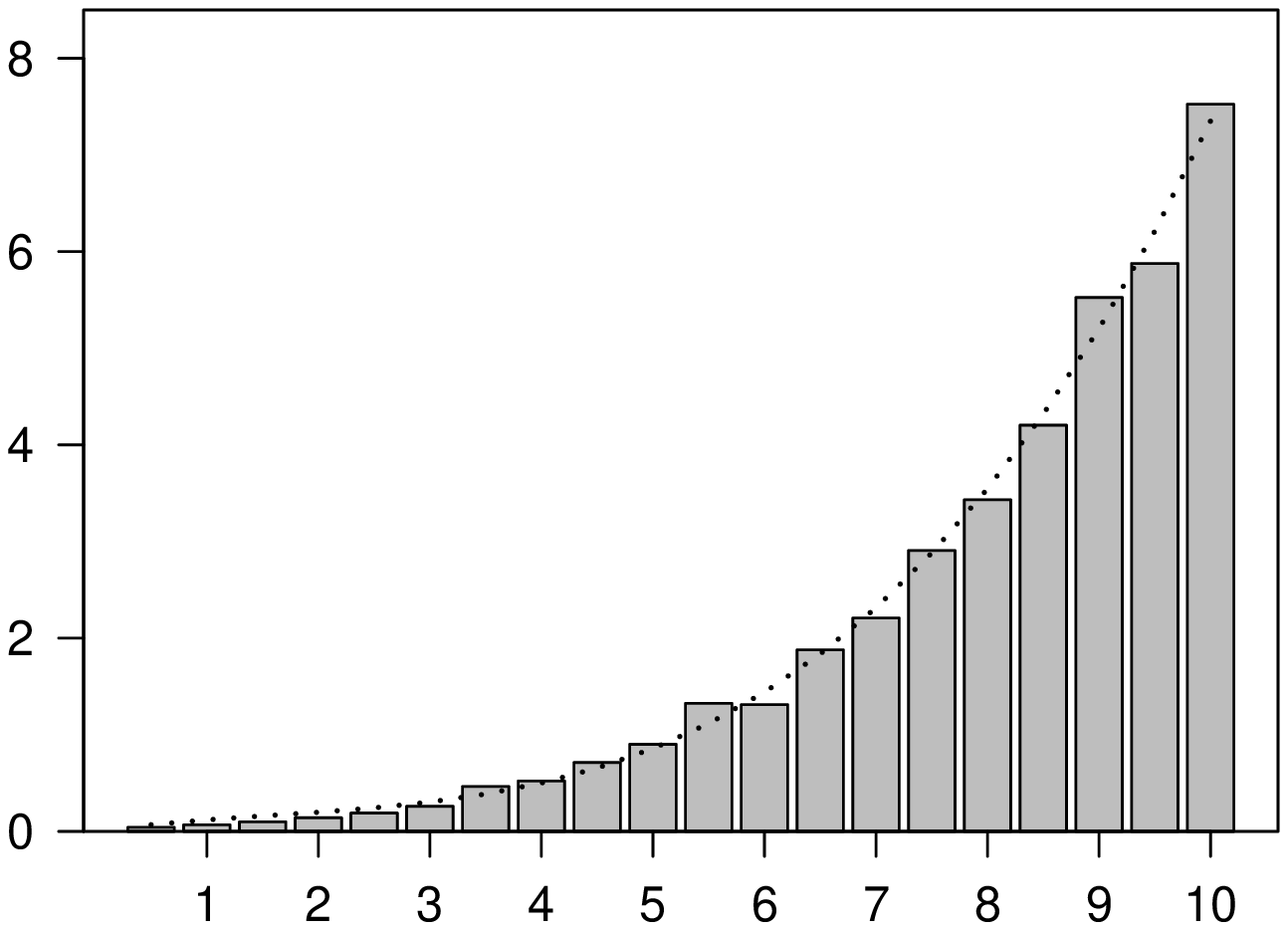}
  \label{fig:nores}}\\
\subfloat[]{\includegraphics[scale=0.42]{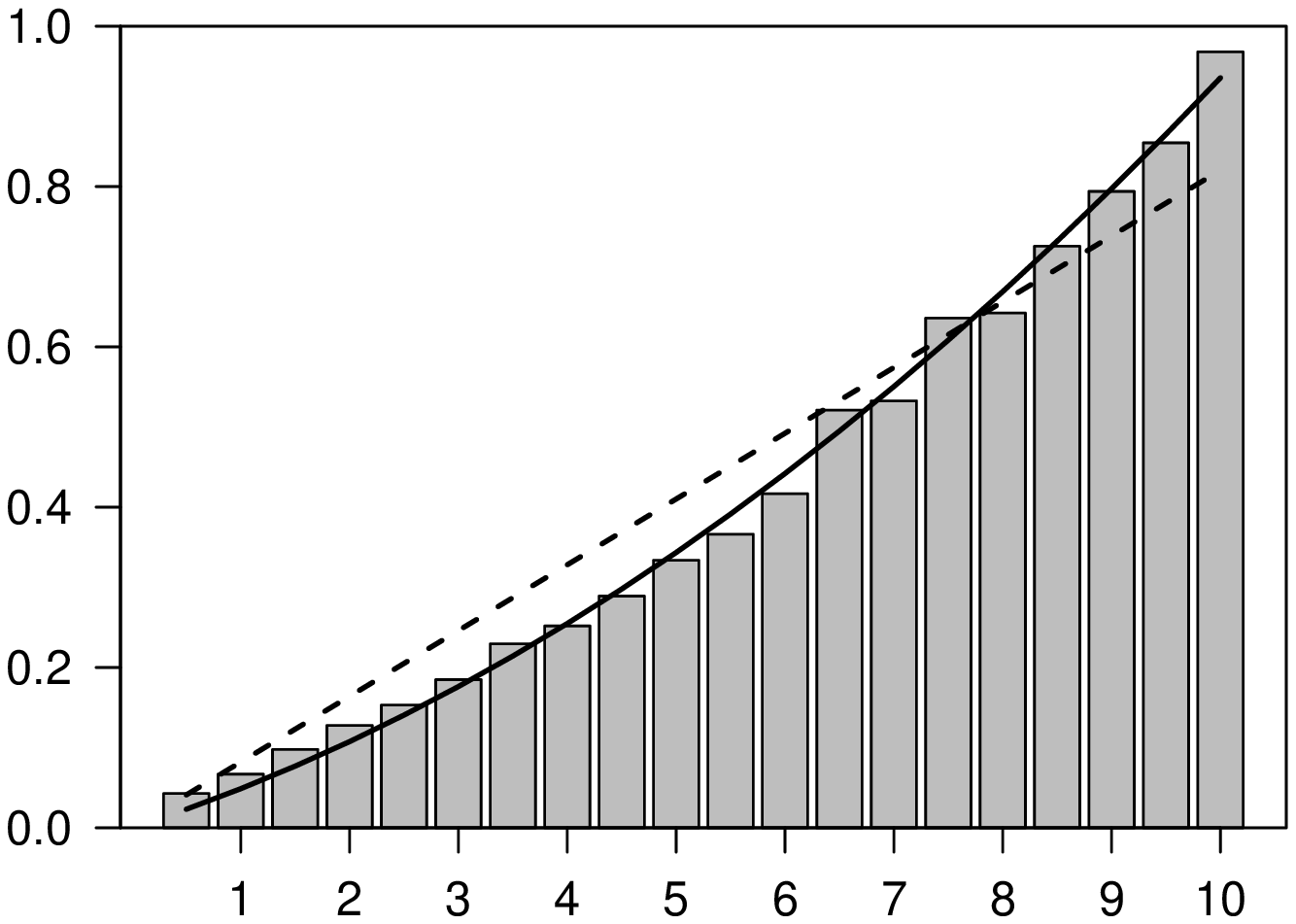}
  \label{fig:res_n}}
\subfloat[]{\includegraphics[scale=0.42]{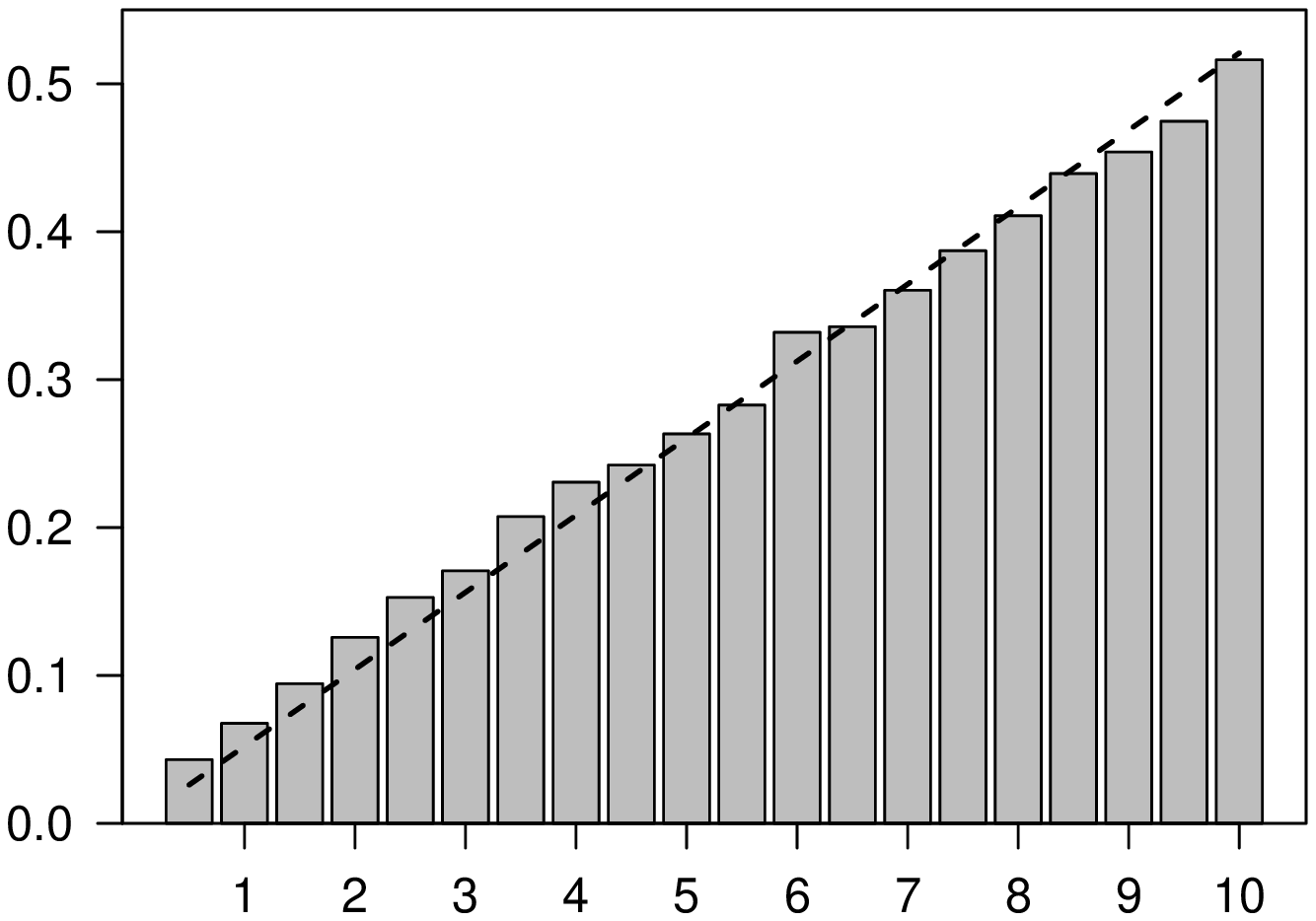}
  \label{fig:res_nsq}}
  \label{fig:inc_dt}
\caption[]{The SV model with $\XX{0}=(1,0)'$ and $(\sigma_1,
    \sigma_2)=(1, 1/2)$. In \subref{fig:inc_dt_1}: MAD of mean vector
    estimators over increasing $T$. In \subref{fig:nores}: RMSE using
    CIS. In \subref{fig:res_n}: RMSE using CIS-R1. In
    \subref{fig:res_nsq}: RMSE using CIS-R2. The dotted line is the
    cubic fit, the solid line is the quadratic fit and the dashed line
    is the linear fit. The results are averages of $10000$
    replications.}
\end{figure}
\footnotetext[1]{Implemented with $\lambda(\tau)=\tau^{-1/2}$}
\footnotetext[2]{Implemented with $\lambda(\tau)=1$}
\footnotetext[3]{Implemented with $\lambda(\tau)=1/2\tau^{-1/2}$}

\subsection{Bivariate CIR model}
\label{sec:cir_model}

In this section we deal with a bivariate Cox-Ingersoll-Ross process
(hereafter denoted by CIR), solution to
\begin{align*}
  \d X_{1,t} &= -\rho_1\lp X_{1,t}-\mu_1\rp\d t + \sigma_1\sqrt{X_{1,t}}\d W_t,\notag \\
  \d X_{2,t} &= -\rho_2\lp X_{2,t}-\mu_2\rp\d t + \sigma_2\sqrt{X_{2,t}}\lp\rho\d W_t + \sqrt{1-\rho^2}\d B_t\rp.
\end{align*}
The instantaneous correlation between the coordinates of the system is
$\rho$, whilst each marginal is a univariate Cox-Ingersoll-Ross
process \citep{cox:ing:ross}. Even though the model can be reduced
into one with unit diffusion matrix, the transformed drift cannot be
written in a gradient form. Therefore, existing exact methods are not
applicable. The CIR process is an example for which the conditions of
Theorem (toDo: REF THRM HERE) do not hold. It is therefore interesting
to examine empirically the behaviour of the CIS algorithm. We perform
two sets of simulations.

The first examines the stability of the weights at different terminal
times for a given set of parameters and fixed initial point. The
distribution of the weights is shown in Figure
\ref{fig:weights_cir}. Notice that, despite the deviation of the
process from Theorem's (toDo: ref here) assumptions, the weights
behave in a stable manner.

The second assesses the performance of the GCIS method for estimating
the transition density of the process and compare it to the
conventional DG estimator \citep{dur:gal:2002} under two
transformation schemes: log transformation and transformation to unit
diffusion matrix (also known as Lamperti transform). The DG method is
employed using $N={\cal O}(M^2)$, where $N$ and $M$ are the number of
trajectories and imputed points per trajectory respectively. This
setting follows from the results of \cite{stramer:yan:asympt} on
optimal allocation of the computing resources. As recommended by the
authors, we set $N=M^2$ as a crude choice. The performance of the
estimators is examined under an increasing sequence of computational
costs $K_i,\, i=1,2,\ldots,7$. The $K_i$s are chosen by selecting
various values for the level of discretisation for the DG approach,
$M_i=2^i$. Thus $K_i = N_iM_i = M_i^3,\, i=1,2,\ldots,7$.

Table \ref{table:cir_tab} summarises the performance of the transition
density estimators for a given initial and terminal point. The GCIS
algorithm performs substantially better than CIS, clearly illustrating
the effect of using a guided proposal process.  For lower values of
the computational cost, the DG approach outperforms CIS but this
behaviour changes as the cost increases. This is not surprising since
DG tries to account for two sources of error, bias arising from the
discretisation and variance due to Monte Carlo error, whereas CIS
allocates all computational effort to the Monte Carlo error. The
Lamperti transform is particularly beneficial to the accuracy of the
CIS and approximately halves its RMSE. Further comparisons were also
performed on different terminal points which yielded similar results
(not shown here).


\begin{figure}
  \centering
  \footnotesize
  \psfrag{t1}{$T=1/10$}
  \psfrag{t2}{$T=1/2$}
  \psfrag{t3}{$T=1$}
  \includegraphics[scale=0.42, angle=-90]{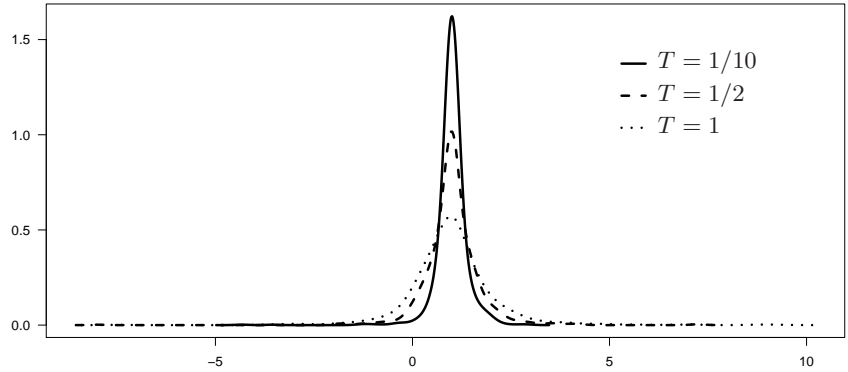}
  \caption{The CIR model with $(\rho_1, \mu_1, \sigma_1,\rho_2, \mu_2,
    \sigma_2, \rho)=(0.6, 2.5, 0.45, 0.3, 3.0, 0.35, 0.5)$, and
    $\XX{0}=(2.5, 3)'$. Kernel density estimates of the weight
    distribution using CIS at different terminal times
    $T$.}\label{fig:weights_cir}
\end{figure}

\newpage
 \footnotetext[4]{$\lambda(\tau)=1/2\tau^{-1/2}$; each
  trajectory was evaluated on average at $3.3$ time points.}
\footnotetext[5]{$\lambda(\tau)=\tau^{-1/2}$; each trajectory was
  evaluated on average at $6.1$ time points.}
\begin{table}
\caption{The CIR model with $(\rho_1, \mu_1, \sigma_1,\rho_2, \mu_2,
  \sigma_2, \rho)=(0.6, 2.5, 0.45, 0.3, 3.0, 0.35, 0.5)$,
  $\XX{0}=\XX{1}=(2.5, 3)'$. Mean, Monte Carlo error and RMSE of transition
  density estimators under
  increasing computational cost $K$. The results are averages of
  $1000$ replications.}
\begin{center}
 \begin{tabular}{llrrrrrrrr}
\hline
& &$K_1$ &$K_2$ &$K_3$ &$K_4$ &$K_5$ &$K_6$ &$K_7$ \\
\hline\\[0.05cm]
CIS \footnotemark[4] &Mean & 0.6692 & 0.6200 & 0.6369 & 0.6380 & 0.6391 & 0.6387 & 0.6386 \\
(log-transform)  &MC Err. & 0.8901 & 0.2904 & 0.1261 & 0.0640 & 0.0227 & 0.0082 & 0.0032 \\
 &RMSE & 0.8907 & 0.2910 & 0.1261 & 0.0640 & 0.0227 & 0.0082 & 0.0032 \\[0.4cm]
GCIS \footnotemark[5]  &Mean & 0.6573 & 0.6324 & 0.6402 & 0.6386 & 0.6389 & 0.6386 & 0.6387 \\
(log-transform) &MC Err. & 0.4425 & 0.1491 & 0.0550 & 0.0201 & 0.0073 & 0.0025 & 0.0009 \\
 &RMSE & 0.4429 & 0.1492 & 0.0550 & 0.0201 & 0.0073 & 0.0025 & 0.0009 \\[0.4cm]
GCIS \footnotemark[5]  &Mean & 0.6361 & 0.6357 & 0.6388 & 0.6386 & 0.6387 & 0.6386 & 0.6386 \\
(lam-transform) &MC Err. & 0.2017 & 0.0704 & 0.0250 & 0.0093 & 0.0033 & 0.0012 & 0.0004 \\
 &RMSE & 0.2017 & 0.0704 & 0.0250 & 0.0093 & 0.0033 & 0.0012 & 0.0004 \\[0.4cm]
DG  &Mean & 0.5231 & 0.5786 & 0.6084 & 0.6235 & 0.6310 & 0.6348 & 0.6367 \\
(log-transform) &MC Err. & 0.0997 & 0.0431 & 0.0167 & 0.0067 & 0.0028 & 0.0013 & 0.0006 \\
 &RMSE & 0.1525 & 0.0739 & 0.0344 & 0.0165 & 0.0081 & 0.0040 & 0.0020 \\[0.4cm]
DG &Mean & 0.5315 & 0.5853 & 0.6110 & 0.6248 & 0.6316 & 0.6351 & 0.6368 \\
(lam-transform) &MC Err. & 0.0873 & 0.0324 & 0.0124 & 0.0049 & 0.0018 & 0.0008 & 0.0004 \\
 &RMSE & 0.1381 & 0.0623 & 0.0303 & 0.0146 & 0.0072 & 0.0036 & 0.0018 \\

\hline
 \end{tabular}
\end{center}
\label{table:cir_tab}
\end{table}





\section{Discussion} \label{sec:discussion}

This paper has introduced the CIS sampler for unbiased simulation of quite general classes of stochastic processes, and demonstrated that the method performs well in practice in a range of scenarios. In doing this we have circumvented the two major limitations of the exact algorithm framework for simulating diffusions, while also providing a framework which can be applicable well beyond the diffusion case.

We have shown that classical ideas due to
Wagner can be used to provide alternative approaches, although it appears that CIS has distinct computational advantages, deriving in particular from its sequential
construction. There are also other, more recent ideas, for unbiasedly sampling from a multivariate diffusions. Using a methodology based on Malliavin weights, \cite{henry2015exact} construct unbiased estimators for diffusion functionals at a fixed time, but it can only be applied when diffusion coefficients are constant. There is also theoretical work on importance
sampling for diffusions by Bally and Kohatsu-Higa  \cite[]{bally2015probabilistic}, from which it may be able to construct
sequential Monte Carlo schemes.
This latter approach is derived from the parametrix method for solving parabolic differential
equations and is based on an iterated integral representation of the
transition density. It is
interesting that interpreting CIS as an iterated integral is not
equivalent to the parametrix and thus, albeit
related, CIS appears to be fundamentally different to this approach. Finally \cite{blanchet2017exact} suggest an approach for unbiasedly simulating multivariate diffusions using rough path techniques, but currently these methods have an expected computational cost which is infinite.

The theoretical stability of the method provides quite delicate and subtle questions, and without carefully constructed proposal stochastic processes, the Monte Carlo weight can easily fail to be stable (and may even fail to be in $L^1$ in some situations). However we give a collection of positive results which provide stability guarantees under quite reasonable and checkable conditions.

In Section \ref{sec:numericals} we have investigated resampling
schemes which are required to provide some robusteness for CIS as the
time interval grows. There are many issues requiring further
investigation here (both theoretically and empirically) including how
best to deal with negative weights, and how to optimally choose times
to update weights in order to minimise the need for resampling. 









\appendix

\section{Appendix}
\def\breq{\notag\\&}
\def\fmrwcc{\mrwcc{}{x}{y}{u}}
\def\fprdcc{\prdcc{}{x}{y}{u}}
\def\fqdder{\qdderf{\theta}{x}{y}{u}}
\def\fqder{\qderf{\theta}{x}{y}{u}}

\paragraph{\it Proof of Lemma \ref{lemma:iweightcc}}

By definition,
\begin{align}\label{eq:cciweight}
  \fmrwcc = 1 + \frac{1}{\lambda(u)\fprdcc}\left[\tro-\prot{\theta}
\right] \fprdcc.
\end{align}
We start by calculating the forward operator of the target process
applied to the proposal density. Specifically, we have that
\begin{align*}
  &\frac{\tro \fprdcc}{\fprdcc} =
  \frac{\sum_{i,j}\frac{\partial^2}{\partial y_i\partial
    y_j}[\dm_{ij}(y)\fprdcc]}{2\fprdcc} -  \frac{\sum_{i}\frac{\partial}{\partial
    y_i}[\bb_i(y) \fprdcc]}{\fprdcc}\breq
=\frac{1}{2}\sum_{i,j}\bigg\{\frac{\partial^2}{\partial y_i\partial
    y_j}\dm_{ij}(y) + \dm_{ij}(y)[\fqdder]_{ij} +
  \frac{\partial\dm_{ij}(y)}{\partial y_i}[\fqder]_j \breq +
  \frac{\partial \dm_{ij}(y)}{\partial y_j}[\fqder]_i\bigg\} - \sum_{i}\lcb \frac{\partial \bb_i(y)}{\partial
    y_i} + \bb_i(y)[\fqder]_i\rcb\breq
=\frac{1}{2}\sum_{i,j}\lcb \frac{\partial^2\dm_{ij}(y)}{\partial y_i\partial
    y_j} + \gamma_{ij}(y)[\fqdder]_{ij}\rcb \breq +
  \sum_{i}\lcb\sum_j\frac{\partial\dm_{ij}(y)}{\partial y_j} -
  \bb_i(y)\rcb[\fqder]_i - \sum_{i}\frac{\partial \bb_i(y)}{\partial
    y_i}.
\end{align*}
Writing the above expression in terms of vectors and matrices, we
obtain that
\begin{align}\label{eq:ccKT}
\frac{\tro \fprdcc}{\fprdcc} &= \frac{1}{2}\lsb\dm(y):\fqdder
 + \dmdder(y):\one_{\nd\times\nd} \rsb \notag\\
&+ \sum_{i}\lcb  [\dmder(y) \one_{\nd\times 1}]_i - \bb_i(y)\rcb[\fqder]_i -
\bder(y)\cdot\one_{\nd\times 1}\notag\\
&=\frac{1}{2}\lsb\dm(y):\fqdder + \dmdder(y):\one_{\nd\times\nd}
\rsb\notag\\ 
&+ \lsb\dmder(y) \one_{\nd\times 1} - b(y)\rsb\cdot\fqder -
\bder(y)\cdot\one_{\nd\times 1}.
\end{align}
The analogous calculation for the operator of the proposed process is
simpler since the drift and diffusion matrix are constant. In
particular, we obtain that
\begin{align}\label{eq:ccKP}
\frac{\pro \fprdcc}{\fprdcc} = \frac{1}{2}\dm(x):\fqdder - \bb(x)\cdot\fqder.
\end{align}
Combining \eqref{eq:cciweight}, \eqref{eq:ccKT} and \eqref{eq:ccKP},
we obtain the formula for the incremental weight as
\begin{align*}
\fmrwcc &= 1 +
  \frac{1}{\lambda(u)}\bigg(\frac{1}{2}\Big\{[\dm(y)-\dm(x)]:\fqdder +
  \dmdder(y):\one_{\nd\times\nd}
  \Big\}\notag\\
&+\lsb\dmder(y) \one_{\nd\times 1} - \bb(y) +
\bb(x)\rsb\cdot\fqder - \bder(y)\cdot\one_{\nd\times 1}\bigg).
\end{align*}
\qed


\bibliographystyle{imsart-number}
\bibliography{reference}


\end{document}